\documentclass[a4paper]{amsart}
\usepackage[T1]{fontenc}
\usepackage{amssymb}
\usepackage{hyperref}

\usepackage{tikz,tikz-qtree}
\usetikzlibrary{backgrounds,arrows,calc,positioning}
\tikzset{every tree node/.append style={circle, inner sep = 1.5pt, fill = black}}

\tikzstyle{cn} = [edge from parent path={(\tikzparentnode.center) -- (\tikzchildnode)}]
\tikzstyle{nc} = [edge from parent path={(\tikzparentnode) -- (\tikzchildnode.center)}]
\tikzstyle{cc} = [edge from parent path={(\tikzparentnode.center) -- (\tikzchildnode.center)}]
\tikzstyle{nn} = [edge from parent path={(\tikzparentnode) -- (\tikzchildnode)}]

\newcommand{\A}{\mathcal{A}}
\newcommand{\B}{\mathcal{B}}
\newcommand{\C}{\mathcal{C}}
\newcommand{\G}{\mathcal{G}}
\newcommand{\N}{\mathbb{N}}
\newcommand{\T}{\mathcal{T}}

\newcommand{\AC}{\mathrm{AC}}
\newcommand{\NC}{\mathrm{NC}}
\newcommand{\TC}{\mathrm{TC}}
\newcommand{\DET}{\mathrm{DET}}
\renewcommand{\P}{\mathrm{P}}
\newcommand{\ALOGTIME}{\mathrm{ALOGTIME}}

	\newtheorem{theorem}{Theorem}
	\numberwithin{theorem}{section}
	
	\newtheorem{definition}[theorem]{Definition}
	\newtheorem{lemma}[theorem]{Lemma}
	\newtheorem{proposition}[theorem]{Proposition}
	\newtheorem{example}[theorem]{Example}

\title{A universal tree balancing theorem}

	\author{Moses Ganardi, Markus Lohrey}
	\address{University of Siegen, Germany}
	\email{\{ganardi,lohrey\}@eti.uni-siegen.de}

\begin{document}
		
\begin{abstract}
	We present a general framework for balancing expressions (terms) in form of so called tree straight-line programs. The latter can be seen as circuits over the free term algebra extended by contexts (terms with a hole) and the operations which insert terms/contexts into contexts. It is shown that for every term one can compute in DLOGTIME-uniform TC$^0$ a tree straight-line program of logarithmic depth and size $O(n/\log n)$. This allows reducing the term evaluation problem over an arbitrary algebra $\mathcal{A}$ to the term evaluation problem over a derived two-sorted algebra $\mathcal{F}(\mathcal{A})$. Several applications are presented: (i) an alternative proof for a recent result by Krebs, Limaye and Ludwig on the expression evaluation problem is given, (ii) it is shown that expressions for an arbitrary (possibly non-commutative) semiring can be transformed in DLOGTIME-uniform TC$^0$ into equivalent circuits of logarithmic depth and size $O(n/\log n)$, and (iii) a corresponding result for regular expressions is shown.
	
\end{abstract}

\maketitle

\section{Introduction}
	
Tree balancing is an important algorithmic technique in the area of parallel algorithms and
circuit complexity. 
The goal is to compute from a given tree that represents an algebraic expression
an equivalent expression of logarithmic depth, which then can be evaluated in time $O(\log n)$ on 
a parallel computation model such as  a PRAM. Equivalence of expressions  usually means that the expressions evaluate to the same
element in an underlying algebraic structure. A widely studied example in this context is the Boolean expression balancing problem,
where the underlying algebraic structure is the Boolean algebra $(\{0,1\}, \vee, \wedge, \neg)$.
Instead of computing an equivalent balanced expression, it is often more natural to compute an equivalent circuit (or dag),
which is a succinct representation of a tree, where identical subtrees are represented only once.
For the Boolean expression balancing problem, 
Spira \cite{Spira71} proved that for every Boolean expression of size $n$ there exists an equivalent Boolean circuit of depth
$O(\log n)$ and size $O(n)$, where the size of a circuit is the number of gates and the depth of a circuit is the length of a longest
path from an input gate to the output gate. Brent \cite{Brent74} extended Spira's theorem to expressions over arbitrary semirings and moreover
improved the constant in the $O(\log n)$ bound. 
Subsequent improvements that mainly concern constant factors can be found in \cite{BonetB94,BshoutyCE95}.

In our recent paper \cite{GHJLN17} we developed a new approach for the construction of logarithmic depth circuits from expressions
that works in two steps. The first step is purely syntactic and is formulated in terms of so called tree straight-line programs.
These are usually defined as  context-free tree grammars that generate a unique tree, see \cite{Lohrey15dlt} for a survey. 
Here we prefer an equivalent definition in terms of circuits over an extension of the free term algebra.
Recall that for a fixed
set $\Sigma$ of ranked function symbols (meaning that every symbol $f \in \Sigma$ has a rank which determines the number
of children of an $f$-labelled node), the free term algebra consists of the set $T(\Sigma)$ of all trees (or terms) over $\Sigma$. Every symbol 
$f \in \Sigma$ of rank $r$ is interpreted by the mapping $(t_1, t_2, \ldots, t_r) \mapsto f(t_1, t_2, \ldots, t_r)$.
Our result from  \cite{GHJLN17} can be formulated over a two-sorted extension $\A(\Sigma)$ of the free term algebra, where the two sorts are
(i) the set $T(\Sigma)$ of all trees over $\Sigma$ and (ii) the set $C(\Sigma)$ of all contexts over $\Sigma$. A context is a tree with
a distinguished leaf that is labelled with a parameter symbol $x$. This allows to do composition of a context $s$ with another
context or tree $t$ by replacing the $x$-labelled leaf in $s$ with $t$; the result is denoted by $s(t)$. 
The algebra $\A(\Sigma)$ is the extension of the free term algebra by the following additional operations:
\begin{itemize}
	\item for all $f \in \Sigma$ of rank $r \geq 1$ and $1 \leq i \leq r$ the $(r-1)$-ary operation
	      $\hat{f}_i : T(\Sigma)^{r-1} \to C(\Sigma)$ with $\hat{f}_i(t_1, \ldots, t_{r-1}) =  f(t_1, \ldots, t_{i-1}, x, t_{i}, \ldots, t_{r-1})$.
	\item the substitution $(s,t) \mapsto s(t)$ for $s \in C(\Sigma)$ and $t \in T(\Sigma)$,
	\item the composition $(s,t) \mapsto s(t)$ for $s,t \in C(\Sigma)$.
\end{itemize}
A tree straight-line program is a circuit over the structure $\A(\Sigma)$ that evaluates to an element of $T(\Sigma)$.
In \cite{GHJLN17}, we proved that from a given tree $t \in T(\Sigma)$ of size $n$ one can construct in logarithmic space
(or, alternatively, in linear time)
a tree straight-line program that evaluates to $t$, has depth $O(\log n)$ and size $O(n/\log n)$ (a simple counting argument shows
that the size bound $O(n/\log n)$ is optimal).
This can be seen as a universal balancing result. It does not refer to an interpretation of the symbols from $\Sigma$
and is purely syntactic. 

The second step of our tree balancing approach from \cite{GHJLN17}  depends on the algebraic structure $\mathcal{A}$ over which the initial
tree $t$ is interpreted. Contexts, i.e. elements of $C(\Sigma)$ cannot be evaluated to elements of $\mathcal{A}$, but they naturally
evaluate to unary linear term functions on $\mathcal{A}$. For many classes of algebraic structures it is possible to represent
these unary functions  by tuples over $\mathcal{A}$. For example, if $\mathcal{A}$ is a (not necessarily commutative) semiring, then
a unary linear term function is an affine mapping $x \mapsto a x b + c$, which can be encoded by the tuple $(a,b,c)$. 

For structures
$\mathcal{A}$ that allow such a representation of unary linear term functions one can transform the 
tree straight-line program obtained from the first step 
into a circuit over the  structure $\mathcal{A}$ that is equivalent to the initial tree $t$. Moreover, the depth (resp., size) of this circuit
is still $O(\log n)$ (resp., $O(n/\log n)$) and the computation of the circuit from the tree straight-line program is of very low
complexity; more precisely it can be accomplished in $\TC^0$ (we always refer to the DLOGTIME-uniform variant of $\TC^0$).

The main complexity bottleneck in the approach sketched above is the logspace bound  in the first step.
Logarithmic space is for many applications too high. A good example
is Buss seminal result stating that the Boolean expression evaluation problem belongs to $\NC^1$ \cite{Bus87}. For this result it is crucial
that the Boolean expression is given as a string (for instance its preorder notation) and not as a tree in pointer representation
(for the latter representation, the evaluation problem is logspace-complete). For Boolean expressions of logarithmic depth, the evaluation
problem can be easily solved in logarithmic time on an alternating Turing machine with a random access tape, which shows membership
in $\ALOGTIME = \NC^1$. Hence, one can obtain an alternative proof of Buss result by showing that Boolean expressions can be 
balanced in $\NC^1$. In this paper, we achieve this goal as a corollary of our main result. 

\subsection{Main results}

We show that the first step of our balancing procedure in \cite{GHJLN17}  can be carried out in $\TC^0$ (using an algorithm different
from the one in \cite{GHJLN17}).
More precisely, we show that from a given expression of size $n$ one can construct in $\TC^0$ a tree straight-line program of depth
$O(\log n)$ and size $O(n/\log n)$. 
The tree straight-line program is given in the {\em extended connection representation}, which is crucial for the applications.
Our approach uses the tree contraction procedure of Abrahamson et al.~\cite{AbrahamsonDKP89}. Buss \cite{Buss93} proved that 
tree contraction can be implemented in $\NC^1$. Elberfeld et al.~\cite{ElberfeldJT12} improved this result to $\TC^0$, thereby showing that one can compute
a tree decomposition of width three and logarithmic height from a given tree in $\TC^0$.  

We follow the ideas from \cite{Buss93,ElberfeldJT12} but have to do several modifications, in particular in order to achieve the size bound $O(n/\log n)$. 
To avoid the usually very technical uniformity considerations, we use the characterization of $\TC^0$ by FOM (first-order logic with the
majority quantifier). In a first step, we show how to define in FOM for a given expression a hierarchical decomposition into subexpressions and contexts,
where the depth of the composition is logarithmic in the size of the expression. In a second step, this decomposition is then transformed into a 
tree straight-line program. To achieve the size bound of $O(n/\log n)$ we use a preprocessing of the tree that is based on the tree contraction
approach from \cite{GMT88}.

We present three applications of our universal balancing result:
\begin{itemize}
	\item We present an alternative (and hopefully  simpler) proof
	      of the main result of \cite{KrebsLL17}, which states that the evaluation problem for expressions over an algebra $\mathcal{A}$
	      can be solved in DLOGTIME-uniform $\mathcal{F}(\A)\text{-}\NC^1$. Here, $\mathcal{F}(\A)$ is the extension of $\mathcal{A}$
	      by $\A[x]$, i.e. all linear unary term functions over $\mathcal{A}$,
	      together with the evaluation operation $\mathcal{A}[x] \times \mathcal{A} \to \mathcal{A}$
	      and the composition operation $\mathcal{A}[x] \times \mathcal{A}[x] \to \mathcal{A}[x]$.
	      The class $\mathcal{F}(\A)\text{-}\NC^1$ is defined by log-depth circuits of polynomial size 
	      over the algebra $\mathcal{F}(\A)$ that may also contain Boolean gates (the interplay between Boolean gates
	      and non-Boolean gates is achieved by multiplexer gates). 
	      We prove the result from \cite{KrebsLL17} as follows: Using our universal balancing theorem, we 
	      transform the input expression over the algebra $\mathcal{A}$ into an equivalent 
	      expression over $\mathcal{F}(\A)$ of logarithmic depth and polynomial size; 
	      see also Theorem~\ref{thm:main-algebra}. This first stage of the computation can be 
	      done in $\TC^0$ and hence in Boolean $\NC^1$. In a second state we use a universal 
	      evaluator circuit for the algebra $\mathcal{F}(\A)$ to evaluate the log-depth expression
	      computed in the first stage. 
	\item We show that for every semiring $\mathcal{S}$, one can transform in $\TC^0$ an arithmetic expression of size $n$ into
	      an equivalent arithmetic circuit of size $O(n / \log n)$ and depth $O(\log n)$.
	\item We show that every regular expression of size $n$ can be transformed in $\TC^0$ into an equivalent circuit (that uses the operators
	      $+$, $\cdot$ and ${}^*$) of size $O(n / \log n)$ and depth $O(\log n)$. This strengthens a result from \cite{GruberH08} stating that 
	      every regular expression of size $n$ has an equivalent regular expression of star height $O(\log n)$
	      (the complexity of this transformation and the total height of the resulting expression is not analyzed in \cite{GruberH08}).
\end{itemize}
Let us finally mention that the idea of evaluating expressions via unary linear term functions  can be also found
in \cite{MillerT87,MillerT97}, where the main goal is to develop optimal parallel circuit and term evaluation algorithms
for the EREW PRAM model.

\section{Preliminaries}
	
\subsection{Terms over algebras}
	
\label{sec-alg}
	
Let $S$ be a finite set of {\em sorts}. An {\em $S$-sorted set} $X$ is a family of sets $\{ X_s \}_{s \in S}$; it is finite
if every $X_s$ is finite.
An {\em $S$-sorted signature} $\Sigma$ is a finite $F$-sorted set $\Sigma$ of {\em function symbols}, where
$F \subseteq (S^* \times S)$ is finite.
If $f$ has sort $(s_1 \cdots s_r,s)$, we simply write $f: s_1 \times \dots \times s_r \to s$ and
call $r \in \N$ the {\em rank} of $f$.
An {\em $S$-sorted algebra} $\A$ over $\Sigma$ consists of an $S$-sorted non-empty {\em domain} $A = \{ A_s \}_{s \in S}$
and	{\em operations} $f^{\A}: A_{s_1} \times \dots \times A_{s_r} \to A_s$ for each function symbol
$f: s_1 \times \dots \times s_r \to s$ in $\Sigma$.
If $|S|=1$, an $S$-sorted signature $\Sigma$ is a {\em ranked alphabet} and an $S$-sorted algebra is simply called an {\em algebra}.

We define the $S$-sorted set $T(\Sigma)$ of {\em terms} over $\Sigma$ inductively:
If $f: s_1 \times \dots \times s_r \to s$ is a function symbol in $\Sigma$  where $r \ge 0$
and $t_1, \dots, t_r$ are terms over $\Sigma$ of sorts $s_1, \dots, s_r$, respectively,
then $f(t_1, \dots, t_r)$ is a term over $\Sigma$ of sort $s$.
A term $t$ over an algebra $\A$ is a term over its signature $\Sigma$ and we will also write 
$T(\A)$ for $T(\Sigma)$. The {\em value} $t^\A \in A$ of a term $t \in T(\A)$ is defined inductively:
If $t = f(t_1, \dots, t_r)$, then $t^\A = f^\A(t_1^\A, \dots, t_r^\A)$.
		
We will also consider terms with a hole, also known as {\em contexts}, which we will need only
for the case $|S|=1$, i.e., a ranked alphabet $\Sigma$. 
Let us fix a special symbol $x \not\in \Sigma$ of rank 0 (the parameter).
A context is obtained from a term $t \in T(\Sigma)$ 
by replacing an arbitrary subterm $t'$ by the symbol $x$. The set of all 
contexts over $\Sigma$ is denoted with $C(\Sigma)$, and if $\A$ is an algebra
over $\Sigma$, we also write $C(\A)$ for $C(\Sigma)$.
Formally, $C(\Sigma)$ is inductively 
defined as follows: $x \in C(\Sigma)$ and if $f \in \Sigma$ has rank $r$, $t_1, \ldots, t_{r-1} \in T(\Sigma)$,
$s \in C(\Sigma)$ and $1 \leq i \leq r$, then $f(t_1, \ldots, t_{i-1},s,t_i, \ldots, t_{r-1}) \in C(\Sigma)$.
Given a context $s$ and a term (resp., context) $t$, we can obtain a term (resp., context) $s(t)$ 
by replacing the unique occurrence of $x$ in $s$ by $t$.
		
In an algebra $\A$, a context $t \in C(\A)$ defines a (unary) {\em linear term function}
$t^{\A} : A \to A$ in the natural way: 
\begin{itemize}
	\item If $t = x$ then $t^\A$ is the identity function.
	\item If $t = f(t_1, \ldots, t_{i-1},s,t_i, \ldots, t_{r-1})$ with $t_1, \ldots, t_{r-1} \in T(\Sigma)$,
	      $s \in C(\Sigma)$, then for every $a \in A$: $t^{\A}(a) = f^\A(t_1^\A, \ldots, t^\A_{i-1},s^\A(a),t^\A_i, \ldots, t^\A_{r-1})$.
\end{itemize}

\subsection{Logical structures and graphs}
	
We will view most objects in this paper as logical structures in order to describe
computations on them by formulas of (extensions of) first-order logic in the framework of descriptive complexity \cite{Imm99}.
A {\em vocabulary} $\tau$ is a tuple $(R_1, \dots, R_k)$
of relation symbols $R_i$ with a certain {\em arity} $r_i \in \N$.
A {\em $\tau$-structure} $\G = (V,R_1^\G, \dots, R_k^\G)$ consists of a non-empty {\em domain} $V = V(\G)$ and
relations $R_i^\G \subseteq V^{r_i}$ for $1 \le i \le k$.
The relation symbols are usually identified with the relations themselves.
All structures in this paper are defined over finite domains $V$, and the size $|V|$ is also denoted by $|\G|$.
	
A {\em graph} $\G$ is a structure of the form $\G  = (V,(E_i)_{1 \le i \le k},(P_a)_{a \in A})$,
where all $E_i$ are binary edge relations and all $P_a$ are unary relations.
The elements of $A$
can be viewed as node labels.
If $\bigcup_{i=1}^k E_i$ is acyclic, $\G$ is called a {\em dag}.
A graph $\G$ is {\em $k$-ordered} if for all $u \in V$ and all $1 \le i \le k$ there exists at most one $v \in V$ with
$(u,v) \in E_i$. If $(u,v) \in E_i$ exists, we call $v$ the {\em $i$-th successor} of $u$.
A {\em tree} $\T = (V,(E_i)_{1 \le i \le k},(P_a)_{a \in A})$ is a graph such that $(V,\bigcup_{i=1}^k E_i)$ is a rooted tree in the usual sense
and the edge relations $E_i$ are pairwise disjoint.
We write $u \preceq_\T v$ if $u$ is an ancestor of $v$ in the rooted tree $(V,\bigcup_i E_i)$.
The {\em depth-first (left-to-right) order} on $V$ defines $v$ to be smaller than $w$ if and only if $v$ is an ancestor of $w$
or there exists a node $u$ and numbers $1 \le i < j \le k$ such that the $i$-th child of $u$ is ancestor of $v$ and
the $j$-th child of $u$ is ancestor of $w$.	
		
\subsection{Circuits}
	
We also make use of a more succinct representation of  terms as defined in Section~\ref{sec-alg}, namely as circuits.
Let $\Sigma$ be an $S$-sorted signature with maximal rank $k$.
A {\em circuit} over $\Sigma$ is a $k$-ordered dag $\C$ whose nodes are called {\em gates}.
The set of gates $V$ is implicitly $S$-sorted.
Each gate $v$ is labelled with a function symbol $f: s_1 \times \dots \times s_r \to s$ from $\Sigma$
such that $v$ has sort $s$ and exactly $r$ successors
where the $i$-th successor of $A$ has sort $s_i$ for all $1 \le i \le r$.
Furthermore $\C$ has a distinguished {\em output gate} $v_{\mathrm{out}} \in V$ (labelled by some special symbol).
The {\em depth} of $\C$ is the maximal path length in $\C$.
The {\em value} of $\C$ over an $S$-sorted algebra $\A$ over $\Sigma$ is defined naturally: One evaluates
all gates of $\C$ bottom-up: If all successor gates of a gate $v$
are evaluated then one can evaluate $v$.
Finally, the value of $\C$ is the value of the output variable $v_{\mathrm{out}}$.
	
In a slightly more general definition, we also allow {\em copy gates} to simplify certain constructions.
A copy gate $v$ (labelled by a special symbol) has exactly one successor $w$
and the value of $v$ is defined as the value of $w$.

\subsection{Circuit complexity and descriptive complexity}

We use standard definitions from circuit complexity, see e.g. \cite{Vol99}. 
The main complexity class used in this paper is DLOGTIME-uniform $\TC^0$,
which is the class of languages $L \subseteq \{0,1\}^*$
recognized by DLOGTIME-uniform circuit families of polynomial size and constant depth
with not-gates and  threshold gates of unbounded fan-in.
If instead of general threshold gates only and-gates and or-gates (again of unbounded fan-in)
are allowed, one obtains DLOGTIME-uniform $\AC^0$.
Analogously, one defines $\AC^0$- and $\TC^0$-computable functions $f: \{0,1\}^* \to \{0,1\}^*$
where the circuit outputs a bit string instead of a single bit.
The definition of DLOGTIME-uniformity can be found in \cite{BIS90}. The precise definition is not 
needed in this paper.

Instead of working with DLOGTIME-uniform circuit families, we will use equivalent
concepts from descriptive complexity based on the logics FO (first-order logic) and 
FOM (first-order logic with the majority quantifier)  \cite{Imm99}.
In this setting we assume that the domain of a structure has the form $\{1, \dots, n\}$.
Furthermore, the vocabulary implicitly contains the binary relations $<$ and $\mathrm{BIT}$,
where $<$ is always interpreted as the natural linear order on $\{1, \dots, n\}$
and $\mathrm{BIT}(i,j)$ is true iff the $j$-th bit of $i$ is 1.
We will not explicitly list these relations when defining structures.
The relations $<$ and $\mathrm{BIT}$ allow to access the bits of elements of the domain
and to do arithmetic manipulation with these elements. In particular, addition and multiplication on the numbers
$\{1, \dots, n\}$ are FO-definable using  $<$ and $\mathrm{BIT}$ \cite[Theorem~1.17]{Imm99}.
Furthermore, if $\Sigma$ is a finite alphabet, in a structure $\G$ of size $n$ we can quantify over sequences
$a_1 \cdots a_s \in \Sigma^*$ of length $s = O(\log n)$ by identifying such sequences by numbers of size $n^{O(1)}$,
or tuples over $V(\G)$ of constant length.
Using a suitable encoding, the BIT-predicate allows us to access each symbol $a_i$ in a FO-formula.
		
An {\em FO-computable function} (or {\em FO-query})
maps a structure $\G$ over some vocabulary to a structure $I(\G)$ over a possible different vocabulary
which is definable in $\G$ by a {\em $d$-dimensional interpretation $I$}
using first-order formulas.
That means, the domain $V(I(\G))$ is an FO-definable subset of $V(\G)^d$ and
each $r$-ary relation in $I(\G)$ is an FO-definable subset of $V(\G)^{d \cdot r}$;
for precise definitions we refer the reader to \cite{Imm99}.

If we additionally allow a majority quantifier in the formulas, we obtain {\em FOM-computable} functions.
Roughly speaking, FOM-logic is the extension of FO-logic by the ability to count. Note that the size of 
$I(\G)$ is polynomially bounded in the size of $\G$.

Notice that, formally one also needs to logically define a linear order $<$ and the BIT-predicate on the output structure $I(\G)$.
For $<$ one can always use the lexicographical order on $V(I(\G)) \subseteq V(\G)^d$
whereas the BIT-predicate might not be definable in first-order logic, cf. \cite[Remark~1.32]{Imm99}.
For example, BIT is FO-definable if the domain formula is valid, i.e. $V(I(\G)) = V(\G)^d$ for all $\G$.
Furthermore, since the BIT-predicate is already FOM-definable from $<$, this technicality vanishes for FOM-computable functions \cite[Theorem~11.2]{BIS90}.
		
The connection between descriptive complexity and circuit complexity is drawn as follows.
A non-empty word $a_1 \cdots a_n \in \{0,1\}^+$ can be viewed as a {\em word structure} $(\{1, \dots, n\}, S)$
where the unary relation $S$ contains those positions $i$ where $a_i = 1$.
A structure $\G$ can be encoded by a bit string $\mathrm{bin}(\G) \in \{0,1\}^*$
in such a way that the conversions between $\G$ and the word structure of $\mathrm{bin}(\G)$ are FO-computable \cite{Imm99}.

It is known that a function $f: \{0,1\}^+ \to \{0,1\}^+$ is FO-computable (respectively, FOM-computable) if and only if
it is computable in DLOGTIME-uniform $\AC^0$ (respectively, DLOGTIME-uniform $\TC^0$).
Hence we can describe $\AC^0$- and $\TC^0$-computations on the binary encoding of a structure
by logical formulas on the structure itself.
	
\section{Representations for trees and dags}
	
It is known that the circuit complexity of algorithmic problems for trees highly depends on the representation of the trees.
For example, for trees given in the standard pointer representation, reachability is complete for deterministic logarithmic space \cite{CoMc87}.
In the {\em ancestor representation}, which is the extension of a tree $\T$ by its ancestor relation $\preceq_\T$,
queries like reachability, least common ancestors and the depth-first order become first-order definable.
Note that a term $t \in T(\Sigma)$ can be represented by a $k$-ordered tree $\T = (V,(E_i)_{1 \le i \le k},(P_a)_{a \in \Sigma})$
where $k$ is the maximal rank of a symbol in $\Sigma$.
Furthermore, each node has a unique label which determines the number of its children, i.e. $\T$ is a {\em ranked tree}.
	
\begin{lemma}[{\cite[Lemma~4.1]{ElberfeldJT12}}]
	\label{lem:term-anc}
	There is an FOM-computable function which transforms a given term $t$ (viewed as a string with opening and closing
	parenthesis)
	into the corresponding labelled ordered tree $\T$ in ancestor representation, and vice versa.
\end{lemma}
	
For ordered dags of logarithmic depth, we propose a representation scheme
which allows to access paths of logarithmic length.
It is similar to the {\em extended connection languages} of circuit families
in the context of uniform circuit complexity \cite{Ruz81}.

A path in a $k$-ordered graph $\G$ can be specified by its start node
and a so called {\em address string} over $\{1, \dots, k\}$.
Formally, for a string $\rho \in \{1, \dots, k\}^*$ 
and a node $u \in V(\G)$ we define the node $\rho(u)$ (it may be undefined) inductively as follows: 
If $\rho = \varepsilon$ then $\rho(u) = u$. Now assume that $\rho = \pi \cdot d$ with $d \in \{1, \dots, k\}$ and the node $v = \pi(u)$ is defined.
Then $\rho(u)$ is the $d$-th successor of $v$, if it is defined, otherwise $\rho(u)$ is undefined.
The {\em extended connection representation}, briefly {\em EC-representation}, of $\G$, denoted by $\mathrm{ec}(\G)$, is the extension of $\G$
by the relation consisting of all so called {\em EC-tuples} $(u, \rho, v)$ where
$u,v \in V(\G)$ and $\rho \in \{1, \dots, k\}^*$ is an address string of length at most $\log_k |\G| -1$ such that $\rho(u) = v$.
Note that there are at most $|\G|$ many such address strings, which therefore can be identified 
with numbers from $1$ to $|\G|$. Hence, 
we can view the set of EC-tuples as a ternary relation over $V(\G)$.
As remarked above, we can access any position of the address string in a first-order formula using the BIT-predicate.
For trees we have:
	
\begin{lemma}
	\label{lem:anc-to-ec}
	There is an FOM-computable function which converts the ancestor representation of a $k$-ordered tree $\T$
	into its EC-representation $\mathrm{ec}(\T)$.
\end{lemma}
	
\begin{proof}
	Let $u,v$ be nodes and $\rho = d_1 \cdots d_{s-1} \in \{1, \dots, k\}^*$ be an address string of logarithmic length.
	Then $(u,\rho,v)$ is an EC-tuple of $\T$ if and only if for all $1 \le i \le s-1$ there exist nodes $v_i, v_{i+1}$
	(which must be unique)
	such that
	\begin{itemize}
		\item $v_{i+1}$ is the $d_i$-th successor of $v_i$,
		\item $|\{ w \in V \mid u \preceq w \preceq v_i \}| = i$, and
		\item $|\{ w \in V \mid v_{i+1} \preceq w \preceq v \}| = s-i$,
	\end{itemize}
	which is FOM-definable using the ancestor relation on $\T$.
\end{proof}
	
Let $\G$ be a $k$-ordered dag and $v_0$ be a node in $\G$. The {\em unfolding} of $\G$ from $v_0$, denoted by $\mathrm{unfold}(\G,v_0)$,
is defined as follows: Its node set is the (finite) set of paths $(v_0, v_1, \dots, v_n) \in V(\G)^+$ starting in $v_0$.
If $v_{n+1}$ is the $i$-th successor of $v_n$ in $G$,
then $(v_0, v_1, \dots, v_{n+1})$ is the $i$-th successor of $(v_0, v_1, \dots, v_n)$ in the unfolding.
The labels of a node $(v_0, v_1, \dots, v_n)$ in the unfolding are the labels of $v_n$ in $G$.
Note that the size of $\mathrm{unfold}(G,v_0)$ can be exponential in the depth of $G$.
	
\begin{lemma}
	\label{lem:unfold}
	For any $c > 0$ there exists an FOM-computable function which, given a $k$-ordered dag $\G$ of size $n$
	and depth $\le c \cdot \log n$ in EC-representation,
	and a node $v_0$, outputs the ancestor representation of the tree $\mathrm{unfold}(\G,v_0)$.
\end{lemma}
	
\begin{proof}
	A node in the unfolding is an address string $\rho \in \{1, \dots, k\}^*$ of length $\le c \cdot \log |G|$,
	such that $\rho(v_0)$ exists. By an FO-formula one can test whether $\rho(v_0)$ exists and also compute this node.
	The $i$-th successor of an address string $\rho$ is the address string $\rho i$.
	The ancestor relation is the prefix relation on the set of address strings. 
\end{proof}
	
In combination with Lemma~\ref{lem:term-anc} this yields:
	
\begin{lemma}
	\label{lem:circuit-unfold}
	For any $c > 0$ there exists an FOM-computable function which, given a circuit $\C$ of size $n$
	and depth $\le c \cdot \log n$ in EC-representation,
	outputs an equivalent term $t$.
\end{lemma}
	
Vice versa, one can compact an ordered tree $\T$ to its {\em minimal dag} $\mathrm{dag}(\T)$. 
It is the up to isomorphism unique smallest dag $\G$ such that
$\T$ is isomorphic to  $\mathrm{unfold}(\G,v)$ for some $v$. One can identify the nodes
of $\mathrm{dag}(\T)$ with the isomorphism classes of the subtrees of $\T$.
	
\begin{lemma}
	\label{lem:min-dag}
	There exists an FOM-computable function which maps a $k$-ordered tree $\T$ in ancestor representation
	to $\mathrm{dag}(\T)$ in EC-representation.
\end{lemma}	
	
\begin{proof}
	Using Lemma~\ref{lem:anc-to-ec} we convert the ancestor representation of $\T$ into its EC-representation.
	With the help of the depth-first order on $V(\T)$, it is FOM-definable
	whether the subtrees rooted in two given nodes $u,v$ are isomorphic.
	For a node $v \in V(\T)$ let $\mathrm{min}(v)$ be the first node (with respect to the built-in order on $V(\T)$)
	such that the subtrees below $v$ and $\mathrm{min}(v)$ are isomorphic. The mapping $\mathrm{min}$ is also FOM-definable.
	Then the node set of $\mathrm{dag}(\T)$ can be identified with $V' = \{ \mathrm{min}(v) \mid v \in V(\T) \}$. A pair 
	$(u',v') \in V' \times V'$ belongs to $E_i^{\mathrm{dag}(\T)}$
	if there exists  $u \in V(\T)$ such that $(u',u) \in E_i^\T$ 
	and $v' = \min(u)$. The set of EC-tuples of $\mathrm{dag}(\T)$ is the set of tuples 
	$(u', \rho, v')$ such that there exists an EC-tuple $(u,\rho,v)$ of $\T$ with $\mathrm{min}(u) = u'$
	and $\mathrm{min}(v) = v'$.
\end{proof}
	
For an arbitrary FOM-computable function $I$ on $k$-ordered graphs, it is not clear whether
the function $\mathrm{ec}(\G) \mapsto \mathrm{ec}(I(\G))$ is FOM-computable as well.
On the other hand, this is possible for so called {\em guarded transductions}.
A  ({\em $m$-dimensional}) {\em connector} has the form
$\gamma: \{1, \ldots, m\} \to \{1, \ldots, k, =\} \times \{1, \ldots, m\}$.
Given a $k$-ordered graph $\G$ and two tuples $\overline{u} = (u_1, \ldots, u_m), \overline{v} = (v_1, \ldots, v_m) \in V(G)^m$,
we say that the connector $\gamma$ {\em connects} $\overline{u}$ to $\overline{v}$ if for all $1 \le j \le m$ the following holds:
\begin{itemize}
	\item If $\gamma(j) = (d,i)$ for some $1 \le d \le k$, then $(u_i,v_j) \in E_d^\G$.
	\item If $\gamma(j) = (=,i)$, then $u_i = v_j$.
\end{itemize}
Notice that $\overline{u}$ and $\gamma$ uniquely determine $\overline{v}$.
Also note that if $k$ and $m$ are constants (as in the lemma below), then a connector 
can be specified with $O(1)$ many bits. Hence, 
a sequence of connectors of length $O(\log |\G|)$ needs $O(\log |\G|)$ bits and can be identified with a tuple 
over $V(\G)$ of fixed length.
	
\begin{lemma} \label{lem:connectors}
	Let $k$ and $m$ be constants.
	Given a $k$-ordered graph $\G$ in EC-representation,
	tuples $\overline{u}, \overline{v} \in V(\G)^m$
	and a sequence $\gamma^{(1)} \cdots \gamma^{(s)}$ of connectors of length $s = O(\log |\G|)$,
	it is FO-definable whether there exists a (necessarily unique) sequence of tuples $\overline{v}^{(1)}, \ldots, \overline{v}^{(s+1)} \in V(\G)^m$
	such that $\overline{v}^{(1)} = \overline{u}$, $\overline{v}^{(s+1)} = \overline{v}$, and
	$\gamma^{(i)}$ connects $\overline{v}^{(i)}$ to $\overline{v}^{(i+1)}$ for all $1 \le i \le s$.
	If so, the tuple sequence is FO-computable in the sense that the $(m+1)$-ary relation 
	$R = \{ (t, \overline{v}^{(t)}) \mid 1 \leq t \leq s+1 \}$
	is FO-computable.
\end{lemma}
	
\begin{proof} Let $\overline{u} = (u_1, \ldots, u_m)$.
	The FO-formula says that for all 
	$2 \le t \le s+1$  there exists a tuple  $\overline{w} = (w_1, \ldots w_m) \in V(\G)^m$ 
	such that for all $1 \le j \le m$ there is 
	a sequence $(j_1,d_1) \cdots (j_{t-1},d_{t-1}) \in (\{1, \ldots, m\} \times \{1, \dots, k, =\})^*$ 
	(which is necessarily unique)
	with the following properties:
	\begin{itemize}
		\item $\gamma^{(i)}(j_{i+1}) = (d_{i},j_{i})$ for all $1 \le i \le t-2$, $\gamma^{(i)}(j) = (d_{t-1},j_{t-1})$ and
		\item  $(u_{j_1}, \pi_t, w_j)$ is an EC-tuple of  $\G$, where the address string
		      $\pi_t \in \{1, \ldots, k\}^*$ is the projection of 
		      $d_1 \cdots d_{t-1}$ to the subalphabet $\{1, \ldots, k\}$.
	\end{itemize}
	Moreover, in case $t = s+1$ we must have $\overline{w} = \overline{v}$. The relation $R$ from the lemma
	contains all tuples $(t,\overline{w})$ and $(1, \overline{u})$.
\end{proof}
	 
A {\em graph transduction} $I$ computes from a $k$-ordered graph $\G$
a $k'$-ordered graph $I(\G)$ whose node set is a subset of $V(\G)^m \times \{1, \dots, c\}$ for some constants $m,c$ 
(we can assume that $\{1, \ldots, c \} \subseteq V(\G)$).
A graph transduction $I$ is {\em guarded} if for every $k$-ordered graph $\G$ and every edge $((\overline{u}, a),(\overline{v},b))$ in $I(\G)$ there exists a connector $\gamma$ which connects $\overline{u}$ to $\overline{v}$.
The idea is that for a given path $(\overline{v}^{(1)},a^{(1)}) (\overline{v}^{(2)},a^{(2)}) \cdots (\overline{v}^{(s)},a^{(s)})$
in $I(\G)$, the connectors $\gamma^{(1)} \gamma^{(2)} \cdots \gamma^{(s-1)}$ describe a forest of paths in $\G$
with its roots in $\overline{v}^{(1)}$, see Figure~\ref{fig:guarded-transduction}.
Based on this forest we can construct the EC-tuples of $I(\G)$ from the EC-tuples of $\G$.
	
\begin{figure}
	\begin{tikzpicture}[->,>=stealth]
		\tikzstyle{p} = [draw, fill = black, circle, inner sep = 0, minimum size = 3pt]
		\node [p] (v11) {};
		\node [p, right = 1.5cm of v11] (v21) {};
		\node [p, right = 1.5cm of v21] (v31) {};
		\node [p, right = 1.5cm of v31] (v41) {};
				
		\node [above = .2cm of v11] {\footnotesize $\overline{v}^{(1)}$};
		\node [above = .2cm of v21] {\footnotesize $\overline{v}^{(2)}$};
		\node [above = .2cm of v31] {\footnotesize $\overline{v}^{(3)}$};
		\node [above = .2cm of v41] {\footnotesize $\overline{v}^{(4)}$};
		
		\node [p, below = .7cm of v11] (v12) {};
		\node [p, below = .7cm of v12] (v13) {};
		\node [p, below = .7cm of v21] (v22) {};
		\node [p, below = .7cm of v22] (v23) {};
		\node [p, below = .7cm of v31] (v32) {};
		\node [p, below = .7cm of v32] (v33) {};
		\node [p, below = .7cm of v41] (v42) {};
		\node [p, below = .7cm of v42] (v43) {};
				
		\node [below right = .2cm and .5cm of v13] {\footnotesize $\gamma^{(1)}$};
		\node [below right = .2cm and .5cm of v23] {\footnotesize $\gamma^{(2)}$};
		\node [below right = .2cm and .5cm of v33] {\footnotesize $\gamma^{(3)}$};
		
		\draw (v11) edge node [above] {\tiny $=$} (v21);
		\draw (v13) edge node [above] {\tiny $1$} (v22);
		\draw (v13) edge node [above] {\tiny $2$} (v23);
				
		\draw (v21) edge node [above] {\tiny $3$} (v31);
		\draw (v22) edge node [above] {\tiny $=$} (v32);
		\draw (v22) edge node [above] {\tiny $2$} (v33);
		
		\draw (v32) edge node [above] {\tiny $1$} (v41);
		\draw (v32) edge node [above] {\tiny $3$} (v42);
		\draw (v33) edge node [above] {\tiny $1$} (v43);
				
		\draw[rounded corners=5pt] ($(v11) + (-.2,.2)$) rectangle ($(v13) + (.2,-.2)$);
		\draw[rounded corners=5pt] ($(v21) + (-.2,.2)$) rectangle ($(v23) + (.2,-.2)$);
		\draw[rounded corners=5pt] ($(v31) + (-.2,.2)$) rectangle ($(v33) + (.2,-.2)$);
		\draw[rounded corners=5pt] ($(v41) + (-.2,.2)$) rectangle ($(v43) + (.2,-.2)$);
			
	\end{tikzpicture}
	\caption{If $I$ is a guarded transduction, a path in $I(\G)$ describes a sequence of connectors.}
	\label{fig:guarded-transduction}
\end{figure}
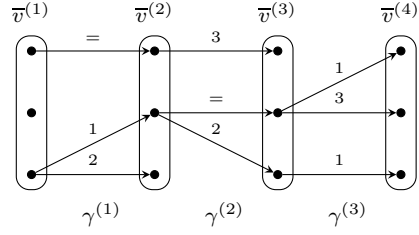
	
\begin{lemma}  \label{lem:guarded}
	For every FOM-computable guarded graph transduction $I$
	there exists an FOM-computable function mapping $\mathrm{ec}(\G)$ to $\mathrm{ec}(I(\G))$ for all $k$-ordered graphs $\G$.
\end{lemma}

\begin{proof}
	It suffices to compute the EC-tuples of $I(\G)$.
	Assume that the the output vocabulary has $k'$ edge relations $E_1, \dots, E_{k'}$.
	Let $(\overline{u},a)$ and $(\overline{v},b)$ be two nodes in $I(\G)$ and $\rho' = d_1' \cdots d_{s}' \in \{1, \dots, k'\}^*$
	be an address string of length at most $\log_{k'} |I(\G)| -1$.
	We claim that one can express by an FO-formula whether $((\overline{u},a),\rho',(\overline{v},b))$ is an EC-tuple of $I(\G)$. This is the case
	if and only if there exist nodes $(\overline{v}^{(1)},a^{(1)}), \dots, (\overline{v}^{(s+1)},a^{(s+1)})$ in $I(\G)$ such that
	$(\overline{v}^{(1)},a^{(1)}) =  (\overline{u},a)$, $(\overline{v}^{(s+1)},a^{(s+1)})=(\overline{v},b)$, and
	\begin{equation} \label{edge-I(G)}
		((\overline{v}^{(t)},a^{(t)}),(\overline{v}^{(t+1)},a^{(t+1)})) \in E_{d_t'}^{I(\G)}
	\end{equation}
	for all $1 \le t \le s$.
	Since $I(\G)$ is $k'$-ordered, these nodes must be unique.
			
	Our FO-formula says that there exists a sequences $a^{(1)} \cdots a^{(s+1)} \in \{1,\ldots,c\}^{s+1}$
	and connectors $\gamma^{(1)}, \ldots, \gamma^{(s)}$ such that there exists  
	a (unique) sequence of tuples $\overline{v}^{(1)}, \ldots, \overline{v}^{(s+1)} \in V(\G)^m$ 
	with $\overline{v}^{(1)} = \overline{u}$, $\overline{v}^{(s+1)} = \overline{v}$, and
	$\gamma^{(i)}$ connects $\overline{v}^{(i)}$ to $\overline{v}^{(i+1)}$ for all $1 \le i \le s$.
	The existence of the sequence $\overline{v}^{(1)}, \ldots, \overline{v}^{(s+1)}$ is expressed using
	Lemma~\ref{lem:connectors}. Moreover, if this sequence exists we can also express whether 
	\eqref{edge-I(G)} holds for all
	$1 \leq t \leq s$ using the FO-computable relation $R$ from Lemma~\ref{lem:connectors}.
\end{proof}
	
Let us remark that in this paper we only need graph transductions where $m = 1$.
The more general definition will be used in a forthcoming paper.
Finally, we will need the following lemma:
	
\begin{lemma}
	\label{lem:copy-gates}
	For any $c > 0$ there exists an FOM-computable function which maps a circuit with copy gates
	of size $n$ and depth $\le c \cdot \log n$ to an equivalent circuit without copy gates with the same depth bound,
	where both circuits are given in EC-representation.
\end{lemma}
	
\begin{proof}
	Let $\C$ be a circuit with copy gates. Let $E$ be the binary relation consisting of all pairs
	$(A,B)$, where $A$ is a copy gate and $B$ is the unique successor of $A$.
	For each copy gate $A \in V$ we define the first non-copy gate on the unique $E$-path starting in $A$, 
	which is first-order definable using the EC-representation.
	By contracting all such paths we can define on all non-copy gates of $\C$ an equivalent circuit $\C'$ without copy-gates.
			
	The EC-tuples of $\C'$ can also be defined in FOM:
	Let $A$ and $B$ be non-copy gates and $\rho' \in \{1, \dots, k\}^*$ be a string of length at most $\log_k |\C'| - 1$
	where $k$ is the maximal rank of a function symbol in $\B$.
	Then $(A,\rho',B)$ is an EC-tuple in $\C'$ if and only if there exists an EC-tuple $(A,\rho,B)$ in $\C$
	such that $\rho'$ is obtained from $\rho$ by omitting those symbols which describe an edge to a non-copy gate on the path $(A,\rho,B)$.
	Formally, we guess a ``bit mask'' $z \in \{0,1\}^*$ with $|z| = |\rho|$ using a single existential quantifier
	and test whether $\rho'$ is obtained from $\rho$ by removing the positions marked with a 0-bit in $z$.
	Then, for each non-empty prefix $\pi$ of $\rho$ we test whether $\pi(A)$ is a non-copy gate
	if and only if $z$ has a 1-bit at position $|\pi|$.
\end{proof}
	
\section{Hierarchical tree definitions}

In the following we will show how to construct a hierarchical definition of a given tree
which has logarithmic depth. Throughout this section all trees are implicitly given in ancestor representation.
The idea is to decompose a tree in a well-nested way into (i) subtrees, (ii) contexts (trees with a hole) and (iii) single nodes.
From such a decomposition, it is easy to derive a tree straight-line program; this will be done in Section~\ref{sec-small-TSLP}.
The advantage of hierarchical decompositions over  tree straight-line programs is that the former perfectly fit into
the descriptive complexity framework: There is a natural representation of
a hierarchical decomposition  by two relations -- a unary one and a binary one --
on the node set of the tree.
	
A {\em pattern} $p$ in a $k$-ordered tree $\T$ is either a single node $v \in V(\T)$, called a {\em subtree pattern},
or a pair of nodes $(v,w) \in V(\T)^2$, called a {\em context pattern}, such that $w$ is a proper descendant of $v$.
A subtree pattern $v$ {\em covers} all descendants of $v$ (including $v$),
whereas a context pattern $(v,w)$ {\em covers} all descendants of $v$ which are not descendants of $w$.
The set of nodes covered by a pattern $p$ is denoted by $V[p]$ and $\T[p]$ is the subtree of $\T$ induced by $V[p]$.
The {\em root} of $p$ is the root of $\T[p]$ and its {\em size} is $|\T[p]|$.
We call $q$ a {\em subpattern} of $p$, denoted by $q \le p$, if $V[q] \subseteq V[p]$,
which partially orders the set of all patterns in a tree.
Note that the root of $\T$ is the largest pattern with respect to $\le$.
Two patterns $p,q$ are {\em disjoint} if $V[p] \cap V[q] = \emptyset$.
A set $P$ of patterns in $\T$ is a {\em hierarchical definition} of $\T$ if
\begin{itemize}
	\item $P$ contains the largest pattern (the root of $\T$), and
	\item $P$ is {\em well-nested},
	      i.e. any two patterns $p,q \in P$ are disjoint
	      or comparable ($p \le q$ or $q \le p$).
	      		
\end{itemize}
The pair $(\T,P)$ is also called a hierarchical definition, which is formally represented as the logical structure
$(\T,P \cap V(\T), P \cap V(\T)^2)$.
	
One can view a hierarchical definition itself as a tree where the patterns are its nodes.
We say that $q \in P$ is a {\em direct subpattern} of $p \in P$ in $P$,
denoted by $q \lessdot p$,
if $q < p$ and there exists no $r \in P$ with $q < r < p$.
The {\em pattern tree} of $P$ is the tree with node set $P$ where the children of a pattern
are its direct subpatterns, ordered by the depth-first order on their roots.
The height of the pattern tree is the {\em depth} of $P$, denoted by $\mathrm{depth}(P)$.
Furthermore, each pattern $p$ in the pattern tree is annotated by its {\em branching tree}, which is 
defined as follows:
The {\em boundary} $\partial p$ of $p$ is $\partial p = V[p] \setminus \bigcup_{q \lessdot p} V[q]$,
i.e. the set of nodes covered by $p$ but not by any of its (direct) subpatterns.
The {\em branching tree} of a pattern $p \in P$ is obtained from $\T[p]$
by contracting the direct subpatterns to single nodes labelled by a special symbol (in order to distinguish them
from boundary nodes), i.e. its node set is
$\partial p \cup \{ q \mid q \lessdot p \}$.
The {\em width} of $P$ is the maximal size of a branching tree of a pattern $p \in P$,
denoted by $\mathrm{width}(P)$.

\tikzstyle{d1} = [line width = 2pt]
\tikzstyle{d2} = [line width = 4pt]
\tikzstyle{d3} = [line width = 6pt]
\tikzstyle{d4} = [line width = 8pt]
\tikzstyle{d5} = [line width = 10pt]
	
\tikzstyle{s1} = [inner sep = 1pt]
\tikzstyle{s2} = [inner sep = 2pt]
\tikzstyle{s3} = [inner sep = 3pt]
\tikzstyle{s4} = [inner sep = 4pt]
\tikzstyle{s5} = [inner sep = 5pt]
	
\definecolor{col1}{rgb}{0.2,0.2,0.8}
\definecolor{col2}{rgb}{0.8,0.2,0.2}
\definecolor{col3}{HTML}{3bc7ce}
\definecolor{col4}{rgb}{0.8,0.2,0.6}
\definecolor{col5}{HTML}{99be1e}
\definecolor{col6}{HTML}{ea953a}
\definecolor{col7}{HTML}{4286e4}

\begin{figure}
	\centering
	\raisebox{-0.5\height}{
		\begin{tikzpicture}[sibling distance = 5pt, level distance = 15pt,
				every tree node/.style={s5, circle, fill = col7},
			edge from parent/.style = {draw, d5, col7, edge from parent path={(\tikzparentnode.center) -- (\tikzchildnode.center)}}]
			\Tree
			[. \node [label={above:\scriptsize $a$}] (a) {};
				[. \node [label={above:\scriptsize \strut $b$}] (b) {};
					[. \node [label={above:\scriptsize \strut $c$}] (c) {};
						[. \node [label={below:\scriptsize \strut $d$}] (d) {}; ]
						[. \node [label={above:\scriptsize $e$}] (e) {};
							[. \node [label={below:\scriptsize \strut $f$}] (f) {}; ]
							[. \node [label={below:\scriptsize \strut $g$}] (g) {}; ]
						]
					]
					[. \node [label={above:\scriptsize $h$}] (h){};
						[. \node [label={below:\scriptsize \strut $i$}] (i) {}; ]
						[. \node [label={below:\scriptsize \strut $j$}] (j) {}; ]
					]
				]
				[. \node [label={above:\scriptsize \strut $k$}] (k) {};
					[. \node [label={above:\scriptsize $l$}] (l) {};
						[. \node [label={above:\scriptsize \strut $m$}] (m) {};
							[. \node [label={below:\scriptsize \strut $n$}] (n) {}; ]
							[. \node [label={below:\scriptsize \strut $o$}] (o) {}; ]
						]
						[. \node [label={below:\scriptsize \strut $p$}] (p) {};
						]
					]
					[. \node [label={above:\scriptsize \strut $q$}] (q) {};
						[. \node [label={below:\scriptsize \strut $r$}] (rr) {}; ]
						[. \node [label={above:\scriptsize \strut $s$}] (s) {};
							[. \node [label={below:\scriptsize \strut $t$}] (t) {}; ]
							[. \node [label={below:\scriptsize \strut $u$}] (u) {}; ]
						]
					]
				]
			]
					
			\tikzstyle{t1} = [circle, draw = col5, fill = col5, s3]
			\node [t1] (1) at (b) {};
			\node [t1] (2) at (c) {};
			\node [t1] (3) at (e) {};
			\node [t1] (4) at (f) {};
			\node [t1] (5) at (g) {};
			\node [t1] (6) at (h) {};
			\node [t1] (7) at (i) {};
			\node [t1] (8) at (j) {};
			\foreach \x/\y in {1/2,2/3,3/4,3/5,1/6,6/7,6/8} {
				\draw (\x.center) edge [d3,col5] (\y.center);
			}
					
			\tikzstyle{t2} = [circle, draw = col4, fill = col4, s2]
			\node [t2] (1) at (c) {};
			\node [t2] (2) at (e) {};
			\node [t2] (3) at (f) {};
			\node [t2] (4) at (g) {};
			\foreach \x/\y in {1/2,2/3,2/4} {
				\draw (\x.center) edge [d2,col4] (\y.center);
			}
			
			\tikzstyle{t3} = [circle, draw = col1, fill = col1, s1]
			\node [t3] (1) at (e) {};
			\node [t3] (2) at (f) {};
			\draw (1.center) edge [d1,col1] (2.center);
			
			\node [t3] (1) at (h) {};
			\node [t3] (2) at (i) {};
			\draw (1.center) edge [d1,col1] (2.center);
			
			\tikzstyle{t4} = [circle, draw = col6, fill = col6, s4]
			\node [t4] (1) at (k) {};
			\node [t4] (2) at (l) {};
			\node [t4] (3) at (m) {};
			\node [t4] (4) at (n) {};
			\node [t4] (5) at (o) {};
			\node [t4] (6) at (p) {};
			\node [t4] (7) at (q) {};
			\node [t4] (8) at (rr) {};
			\node [t4] (9) at (s) {};
			\node [t4] (10) at (t) {};
			\foreach \x/\y in {1/2,2/3,3/4,3/5,2/6,1/7,7/8,7/9,9/10} {
				\draw (\x.center) edge [d4,col6] (\y.center);
			}
					
			\tikzstyle{t5} = [circle, draw = col3, fill = col3, s3]
			\node [t5] (0) at (k) {};
			\node [t5] (1) at (l) {};
			\node [t5] (2) at (m) {};
			\node [t5] (3) at (n) {};
			\node [t5] (4) at (o) {};
			\node [t5] (5) at (p) {};
			\foreach \x/\y in {0/1,1/2,2/3,2/4,1/5} {
				\draw (\x.center) edge [d3,col3] (\y.center);
			}
					
			\tikzstyle{t6} = [circle, draw = col2, fill = col2, s2]
			\node [t6] (1) at (l) {};
			\node [t6] (2) at (m) {};
			\node [t6] (3) at (n) {};
			\node [t6] (4) at (p) {};
			\foreach \x/\y in {1/2,2/3,1/4} {
				\draw (\x.center) edge [d2,col2] (\y.center);
			}
					
			\node [t3] (1) at (m) {};
			\node [t3] (2) at (n) {};
			\draw (1.center) edge [d1,col1] (2.center);
			
			\node [t3] (1) at (s) {};
			\node [t3] (2) at (t) {};
			\draw (1.center) edge [d1,col1] (2.center);
		\end{tikzpicture}
	}
	\raisebox{-0.5\height}{
		\tikzstyle{branching-tree} = [nn, every tree node/.style={draw = none, circle, fill = white, inner sep = 0, minimum size = 0pt},
		sibling distance = 8pt, level distance = 10pt]
		\tikzstyle{direct-subpattern} = [minimum size = 5pt]
		\begin{tikzpicture}[every node/.style={draw}]
			\node [draw = col7] (a) { \begin{tikzpicture}[branching-tree]
				\Tree
				[. \node {\footnotesize $a$};
					[. \node {\footnotesize $*$};
						[. \node {\footnotesize $d$}; ]
					]
					[. \node {\footnotesize $*$};
						[. \node {\footnotesize $u$}; ]
					]
				]
				\end{tikzpicture}
			};
			\node [draw = col5] (b) [below left = 0pt and 5pt of a] { \begin{tikzpicture}[branching-tree]
				\Tree
				[. \node {\footnotesize $b$};
					[. \node {\footnotesize $*$}; ]
					[. \node {\footnotesize $*$};
						[. \node {\footnotesize $j$}; ]
					]
				]
				\end{tikzpicture}
			};
			\node [draw = col6] (c) [below right = 0pt and 5pt of a] { \begin{tikzpicture}[branching-tree]
				\Tree
				[. \node {\footnotesize $*$};
					[. \node {\footnotesize $q$};
						[. \node {\footnotesize $r$}; ]
						[. \node {\footnotesize $*$}; ]
					]
				]
								
				\end{tikzpicture}
			};
			\node (b1) [draw = none, below left = 8pt of b] {};
			\node (b2) [draw = none, below right = 8pt of b] {};
			\node (c1) [draw = none, below left = 8pt of c] {};
			\node (c2) [draw = none, below right = 8pt of c] {};
				
			\draw (a.south) ++ (-10pt,0) edge (b);
			\draw (a.south) ++ (10pt,0) edge (c);
			\draw (b.south) ++ (-10pt,0) edge (b1);
			\draw (b.south) ++ (10pt,0) edge (b2);
			\draw (c.south) ++ (-10pt,0) edge (c1);
			\draw (c.south) ++ (10pt,0) edge (c2);
		\end{tikzpicture}
	}
	\caption{A hierarchical definition with its pattern tree.
		Each pattern in the pattern tree is labelled by its branching tree.
	The symbol $*$ represents a direct subpattern.}
	\label{fig:hier-def}
\end{figure}
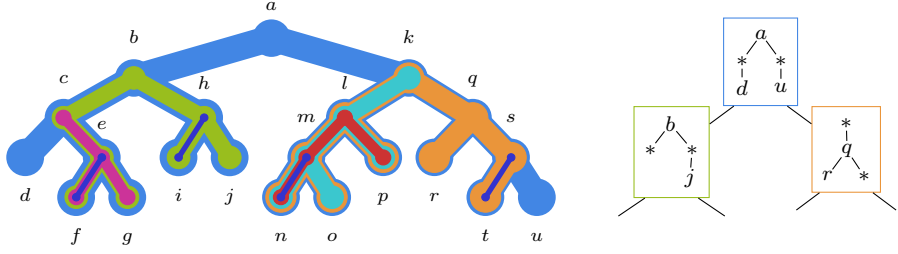

\begin{example}
	Figure~\ref{fig:hier-def} shows an example of a hierarchical definition $P$
	and a top part of its pattern tree, which has height 4.
	The largest pattern in blue has two direct subpatterns
	and its boundary is $\{ a,d,u \}$.
	Hence the branching tree of the largest pattern has size 5,
	which is also the width of $P$.
\end{example}

\subsection{Hierarchical definitions via tree contraction}
	
This section is the core of the paper. 
Using the tree construction technique of 
Abrahamson et al.~\cite{AbrahamsonDKP89} we construct a 
hierarchical definition for a given binary tree. Here, a binary tree is a $2$-ordered
tree $\T = (V,E_1, E_2, (P_a)_{a \in A})$ where every node $u \in V$ is either a leaf
(i.e., there is no $v$ with $(u,v) \in E_1 \cup E_2$) or has a left and a right child (i.e., there exist $v_1, v_2 \in V$ with
$(u,v_1) \in E_1$ and $(u,v_2) \in E_2$). Such trees are also called full binary trees. 
	
The unary relations $P_a$ that
define the node labels are not important in this section and can be completely ignored; only in Section~\ref{sec-n/logn} the number
of node labels will be relevant.

Let $\T$ be a binary tree with at least two leaves.
The basic operation of tree contraction is called {\em prune-and-bypass}.
Let $w$ be a leaf node, $v$ its parent node and $u$ be the parent node of $v$.
Applying the prune-and-bypass operation to $w$ means: both $v$ and $w$ are removed
and the sibling $w'$ of $w$ becomes a new child of $u$.
We say that the edges $(u,v)$, $(v,w)$ and $(v,w')$ are {\em involved} in this prune-and-bypass step.
In our definition the operation can only be applied to leaves
of depth at least 2 so that the root is never removed.

To verify the correctness of our (parallel) tree contraction algorithm, we first present a sequential tree-contraction algorithm.
A pattern $p$ in $\T$ is {\em hidden} in a context pattern $(u,v)$ if $p$ is a subpattern of $(u,v)$ but does not cover $u$.
Starting with $P_0 = \emptyset$ and $\T_0 = \T$, we maintain the following invariants:
(1) each pattern $p \in P_i$ is hidden in some edge of $\T_i$ (interpreted as patterns in $\T$)
and (2) $P_i$ is well-nested.
We obtain $\T_{i+1}$ from $\T_i$ by pruning-and-bypassing an arbitrary leaf node $w$ in $\T_i$ (of depth at least 2).
Let $u$ be the grandparent node of $w$ and $w'$ be the sibling of $w$ in $\T_i$.
The {\em contraction pattern} $p$ formed in this prune-and-bypass step is the maximal subpattern $p$
which is hidden in $(u,w')$.
It is the pattern $p = (u',w')$ where $u'$ is the child of $u$ that belongs to the path in $\T$ from $u$ down to $w$.
We add $p$ to $P_i$ to obtain $P_{i+1}$.
	
Clearly, property (1) is preserved because the edge $(u,w')$ is introduced in $\T_{i+1}$
and all patterns which are hidden in some involved edge in $\T_i$ are hidden in the edge $(u,w')$ in $\T_{i+1}$.
By property (1) every pattern $q \in P_i$ which intersects the contraction pattern $p$ is hidden in one of the three edges
involved in the prune-and-bypass operation, therefore $q$ is a subpattern of $p$.
This proves that $P_{i+1}$ is indeed well-nested.
By adding the largest pattern in $\T$ to any set $P_i$, we clearly obtain a hierarchical definition for $\T$.
	
Now we proceed with the parallel tree-contraction algorithm.
Notice that we can apply the prune-and-bypass operation to a set of leaves in parallel
if no edge is involved in more than one prune-and-bypass operation.
We apply the prune-and-bypass operation only 
to {\em internal} leaves, i.e. leaves which are not the left- or the right-most leaf in the tree.
This implies that leaves which are children of the root node are not pruned, i.e., every pruned leaf
has a grandparent as required above.\footnote{Elberfeld et al.~\cite{ElberfeldJT12} enforce this
	by adding at the very beginning a fresh root with a fresh leaf as its left child, and the original tree as its right subtree. Here,
we want to avoid adding new nodes to the tree.}
Let $\T_0$ be the input tree $\T$ with $n$ internal leaves and hence $n+2$ leaves and $2(n+2)-1$ nodes.
We label the internal leaves by the numbers $1, \dots, n$ from left to right.
It may be helpful for the reader to think of the leaf numbers in their binary encodings.
We construct a sequence of trees $\T_0, \dots, \T_m$ as follows.
\begin{itemize}
	\item If $\T_{2i}$ has size two, the algorithm terminates.
	\item If $\T_{2i}$ has at least one internal leaf, we prune-and-bypass all internal leaves in $\T_{2i}$ with an odd number that are left children to obtain the tree $\T_{2i+1}$.
	      Then we prune-and-bypass all internal leaves in $\T_{2i+1}$ with an odd number that are right children
	      and relabel the remaining internal leaves (divide leaf number by 2) to obtain the tree $\T_{2i+2}$.
\end{itemize}
Notice that $\T_{2i}$ contains exactly those internal leaves whose number in $\T_0$ is divided by $2^i$.
Hence the algorithm terminates after $m = 2 (\lfloor \log_2 n \rfloor + 1)$ rounds.
In Figure~\ref{fig:contract-compress} we illustrate the tree contraction algorithm.
The leaves which are pruned and bypassed are colored together with their parent nodes and the involved edges.
	
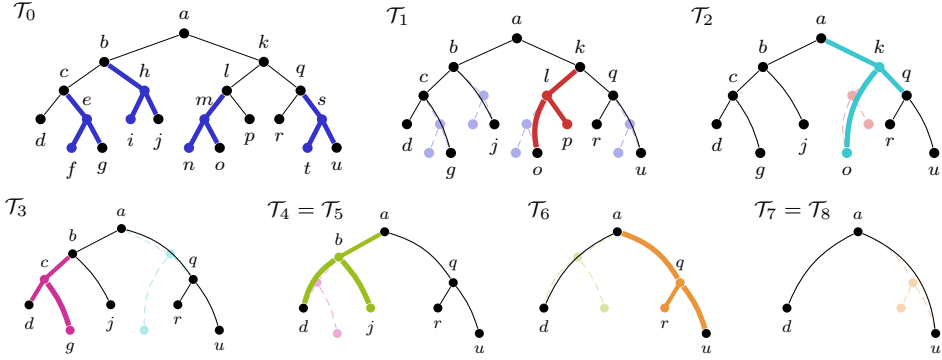
\begin{figure}
			
	\tikzset{sibling distance = 1.5pt, level distance = 12pt}
	\tikzstyle{inv} = [draw=none,fill=none]
			
	\centering
	{
		\begin{tikzpicture}[nn,scale=0.9]
			\Tree
			[. \node [label=\footnotesize $a$] (a) {};
				[. \node [label=\footnotesize $b$] {};
					[. \node [label=\footnotesize $c$] {};
						[. \node [label=below:\footnotesize $d$] {}; ]
						\edge [nc, d1, col1];
						[. \node [label=\footnotesize $e$, col1] {};
							\edge [cc, d1, col1];
							[. \node [label=below:\footnotesize $f$, col1] {}; ]
							\edge [cn, d1, col1];
							[. \node [label=below:\footnotesize $g$] {}; ]
						]
					]
					\edge [nc, d1, col1];
					[. \node [label=\footnotesize $h$, col1] {};
						\edge [cc, d1, col1];
						[. \node [label=below:\footnotesize $i$, col1] {}; ]
						\edge [cn, d1, col1];
						[. \node [label=below:\footnotesize $j$] {}; ]
					]
				]
				[. \node [label=\footnotesize $k$] {};
					[. \node [label=\footnotesize $l$] {};
						\edge [nc, d1, col1];
						[. \node [label=\footnotesize $m$, col1] {};
							\edge [cc, d1, col1];
							[. \node [label=below:\footnotesize $n$, col1] {}; ]
							\edge [cn, d1, col1];
							[. \node [label=below:\footnotesize $o$] {}; ]
						]
						[. \node [label=below:\footnotesize $p$] {}; ]
					]
					[. \node [label=\footnotesize $q$] {};
						[. \node [label=below:\footnotesize $r$] {}; ]
						\edge [nc, d1, col1];
						[. \node [label=\footnotesize $s$, col1] {};
							\edge [cc, d1, col1];
							[. \node [label=below:\footnotesize $t$, col1] {}; ]
							\edge [cn, d1, col1];
							[. \node [label=below:\footnotesize $u$] {}; ]
						]
					]
				]
			]
			\node [above left = 0pt and 50pt of a] {\footnotesize $\T_0$};
		\end{tikzpicture}
	}
	{
		\begin{tikzpicture}[nn,scale=0.9]
			\tikzstyle{hide} = [draw = col1!40, densely dashed]
			\tikzstyle{bypass} = [draw = none, fill = col1!40]
			\Tree
			[. \node [label=\footnotesize $a$] (a) {};
				[. \node [label=\footnotesize $b$] (b) {};
					[. \node [label=\footnotesize $c$] (c) {};
						[. \node [label=below:\footnotesize $d$] {}; ]
						\edge [hide];
						[. \node [bypass] {};
							\edge [hide];
							[. \node [bypass] {}; ]
							\edge [hide];
							[. \node [label=below:\footnotesize $g$] (g) {}; ]
						]
					]
					\edge [hide];
					[. \node [bypass] {};
						\edge [hide];
						[. \node [bypass] {}; ]
						\edge [hide];
						[. \node [label=below:\footnotesize $j$] (j) {}; ]
					]
				]
				[. \node [label=\footnotesize $k$] {};
					\edge [nc, d1, col2];
					[. \node [label=above:\footnotesize $l$, col2] (l) {};
						\edge [hide];
						[. \node [bypass] {};
							\edge [hide];
							[. \node [bypass] {}; ]
							\edge [hide];
							[. \node [label=below:\footnotesize $o$] (o) {}; ]
						]
						\edge [cc, d1, col2];
						[. \node [label=below:\footnotesize $p$, col2] {}; ]
					]
					[. \node [label=\footnotesize $q$] (q)  {};
						[. \node [label=below:\footnotesize $r$] {}; ]
						\edge [hide];
						[. \node [bypass] {};
							\edge [hide];
							[. \node [bypass] {}; ]
							\edge [hide];
							[. \node [label=below:\footnotesize $u$] (u) {}; ]
						]
					]
				]
			]
			\node [above left = 0pt and 35pt of a] {\footnotesize $\T_1$};
			\draw (c) edge [bend left = 10] (g);
			\draw (b) edge [bend left = 10] (j);
			\draw (l) edge [cc, d1, col2, bend right = 20] (o);
			\draw (q) edge [bend left = 10] (u);
		\end{tikzpicture}
	}
	{
		\begin{tikzpicture}[nn,scale=0.9]
			\tikzstyle{hide} = [draw = col2!40, densely dashed]
			\tikzstyle{bypass} = [draw = none, fill = col2!40]
			\Tree
			[. \node [label=\footnotesize $a$] (a) {};
				[. \node [label=\footnotesize $b$] (b) {};
					[. \node [label=\footnotesize $c$] (c) {};
						[. \node [label=below:\footnotesize $d$] {}; ]
						\edge [inv];
						[. \node [inv] {};
							\edge [inv];
							[. \node [inv] {}; ]
							\edge [inv];
							[. \node [label=below:\footnotesize $g$] (g) {}; ]
						]
					]
					\edge [inv];
					[. \node [inv] {};
						\edge [inv];
						[. \node [inv] {}; ]
						\edge [inv];
						[. \node [label=below:\footnotesize $j$] (j) {}; ]
					]
				]
				\edge [nc, d1, col3];
				[. \node [label=\footnotesize $k$, col3] (k) {};
					\edge [hide];
					[. \node (l) [bypass] {};
						\edge [inv];
						[. \node [inv] {};
							\edge [inv];
							[. \node [inv] {}; ]
							\edge [inv];
							[. \node [label=below:\footnotesize $o$, col3] (o) {}; ]
						]
						\edge [hide];
						[. \node [bypass] {}; ]
					]
					\edge [cn, d1, col3];
					[. \node [label=\footnotesize $q$] (q)  {};
						[. \node [label=below:\footnotesize $r$] {}; ]
						\edge [inv];
						[. \node [inv] {};
							\edge [inv];
							[. \node [inv] {}; ]
							\edge [inv];
							[. \node [label=below:\footnotesize $u$] (u) {}; ]
						]
					]
				]
			]
			\node [above left = 0pt and 35pt of a] {\footnotesize $\T_2$};
			\draw (c) edge [bend left = 10] (g);
			\draw (b) edge [bend left = 10] (j);
			\draw (l) edge [hide, bend right = 20] (o);
			\draw (k) edge [cc, d1, col3, bend right = 20] (o);
			\draw (q) edge [bend left = 10] (u);
		\end{tikzpicture}
	}
	\hspace{5pt}
	\tikzset{sibling distance = 2pt}
	{
		\begin{tikzpicture}[nn,scale=0.8]
			\tikzstyle{hide} = [draw = col3!40, densely dashed]
			\tikzstyle{bypass} = [draw = none, fill = col3!40]
			\Tree
			[. \node [label=\footnotesize $a$] (a) {};
				[. \node [label=\footnotesize $b$] (b) {};
					\edge [nc, d1, col4];
					[. \node [label=\footnotesize $c$, col4] (c) {};
						\edge [cn, d1, col4];
						[. \node [label=below:\footnotesize $d$] {}; ]
						\edge [inv];
						[. \node [inv] {};
							\edge [inv];
							[. \node [inv] {}; ]
							\edge [inv];
							[. \node [label=below:\footnotesize $g$, col4] (g) {}; ]
						]
					]
					\edge [inv];
					[. \node [inv] {};
						\edge [inv];
						[. \node [inv] {}; ]
						\edge [inv];
						[. \node [label=below:\footnotesize $j$] (j) {}; ]
					]
				]
				\edge [hide];
				[. \node [bypass] (k) {};
					\edge [inv];
					[. \node (l) [inv] {};
						\edge [inv];
						[. \node [inv] {};
							\edge [inv];
							[. \node [inv] {}; ]
							\edge [inv];
							[. \node [bypass] (o) {}; ]
						]
						\edge [inv];
						[. \node [inv] {}; ]
					]
					\edge [hide];
					[. \node [label=\footnotesize $q$] (q)  {};
						[. \node [label=below:\footnotesize $r$] {}; ]
						\edge [inv];
						[. \node [inv] {};
							\edge [inv];
							[. \node [inv] {}; ]
							\edge [inv];
							[. \node [label=below:\footnotesize $u$] (u) {}; ]
						]
					]
				]
			]
			\node [above left = 0pt and 30pt of a] {\footnotesize $\T_3$};
			\draw (c) edge [cc, d1, col4, bend left = 10] (g);
			\draw (b) edge [bend left = 10] (j);
			\draw (a) edge [bend left = 20] (q);
			\draw (k) edge [hide, bend right = 20] (o);
			\draw (q) edge [bend left = 10] (u);
		\end{tikzpicture}
	}
	{
		\begin{tikzpicture}[nn,scale=0.8]
			\tikzstyle{hide} = [draw = col4!40, densely dashed]
			\tikzstyle{bypass} = [draw = none, fill = col4!40]
			\Tree
			[. \node [label=\footnotesize $a$] (a) {};
				\edge [nc, d1, col5];
				[. \node [label=above:\footnotesize $b$,col5] (b) {};
					\edge [hide];
					[. \node [bypass] (c) {};
						\edge [hide];
						[. \node [label=below:\footnotesize $d$] (d) {}; ]
						\edge [inv];
						[. \node [inv] {};
							\edge [inv];
							[. \node [inv] {}; ]
							\edge [inv];
							[. \node [bypass] (g) {}; ]
						]
					]
					\edge [inv];
					[. \node [inv] {};
						\edge [inv];
						[. \node [inv] {}; ]
						\edge [inv];
						[. \node [label=below:\footnotesize $j$,col5] (j) {}; ]
					]
				]
				\edge [inv];
				[. \node [inv] (k) {};
					\edge [inv];
					[. \node (l) [inv] {};
						\edge [inv];
						[. \node [inv] {};
							\edge [inv];
							[. \node [inv] {}; ]
							\edge [inv];
							[. \node [inv] (o) {}; ]
						]
						\edge [inv];
						[. \node [inv] {}; ]
					]
					\edge [inv];
					[. \node [label=\footnotesize $q$] (q)  {};
						[. \node [label=below:\footnotesize $r$] {}; ]
						\edge [inv];
						[. \node [inv] {};
							\edge [inv];
							[. \node [inv] {}; ]
							\edge [inv];
							[. \node [label=below:\footnotesize $u$] (u) {}; ]
						]
					]
				]
			]
			\node [above left = 0pt and 10pt of a] {\footnotesize $\T_4=\T_5$};
			\draw (c) edge [hide, bend left = 10] (g);
			\draw (b) edge [nc, d1, col5, bend left = 10] (j);
			\draw (a) edge [bend left = 20] (q);
			\draw (q) edge [bend left = 10] (u);
			\draw (b) edge [cn, d1, bend right = 20, col5] (d);
		\end{tikzpicture}
	}
	{
		\begin{tikzpicture}[nn,scale=0.8]
			\tikzstyle{hide} = [draw = col5!40, densely dashed]
			\tikzstyle{bypass} = [draw = none, fill = col5!40]
			\Tree
			[. \node [label=\footnotesize $a$] (a) {};
				\edge [hide];
				[. \node [bypass] (b) {};
					\edge [inv];
					[. \node [inv] (c) {};
						\edge [inv];
						[. \node [label=below:\footnotesize $d$] (d) {}; ]
						\edge [inv];
						[. \node [inv] {};
							\edge [inv];
							[. \node [inv] {}; ]
							\edge [inv];
							[. \node [inv] (g) {}; ]
						]
					]
					\edge [inv];
					[. \node [inv] {};
						\edge [inv];
						[. \node [inv] {}; ]
						\edge [inv];
						[. \node [bypass] (j) {}; ]
					]
				]
				\edge [inv];
				[. \node [inv] (k) {};
					\edge [inv];
					[. \node (l) [inv] {};
						\edge [inv];
						[. \node [inv] {};
							\edge [inv];
							[. \node [inv] {}; ]
							\edge [inv];
							[. \node [inv] (o) {}; ]
						]
						\edge [inv];
						[. \node [inv] {}; ]
					]
					\edge [inv];
					[. \node [label=\footnotesize $q$, col6] (q)  {};
						\edge [cc, d1, col6];
						[. \node [label=below:\footnotesize $r$, col6] {}; ]
						\edge [inv];
						[. \node [inv] {};
							\edge [inv];
							[. \node [inv] {}; ]
							\edge [inv];
							[. \node [label=below:\footnotesize $u$] (u) {}; ]
						]
					]
				]
			]
			\node [above left = 0pt and 20pt of a] {\footnotesize $\T_6$};
			\draw (b) edge [hide, bend left = 10] (j);
			\draw (a) edge [nc, d1, col6, bend left = 20] (q);
			\draw (q) edge [cn, d1, col6, bend left = 10] (u);
			\draw (b) edge [hide, bend right = 20] (d);
			\draw (a) edge [bend right = 20] (d);
		\end{tikzpicture}
	}
	{
		\begin{tikzpicture}[nn,scale=0.8]
			\tikzstyle{hide} = [draw = col6!40, densely dashed]
			\tikzstyle{bypass} = [draw = none, fill = col6!40]
			\Tree
			[. \node (a) [label=\footnotesize $a$] {};
				\edge [inv];
				[. \node [inv] (b) {};
					\edge [inv];
					[. \node [inv] (c) {};
						\edge [inv];
						[. \node [label=below:\footnotesize $d$] (d) {}; ]
						\edge [inv];
						[. \node [inv] {};
							\edge [inv];
							[. \node [inv] {}; ]
							\edge [inv];
							[. \node [inv] (g) {}; ]
						]
					]
					\edge [inv];
					[. \node [inv] {};
						\edge [inv];
						[. \node [inv] {}; ]
						\edge [inv];
						\edge [inv];
						[. \node [inv] (j) {}; ]
					]
				]
				\edge [inv];
				[. \node [inv] (k) {};
					\edge [inv];
					[. \node (l) [inv] {};
						\edge [inv];
						[. \node [inv] {};
							\edge [inv];
							[. \node [inv] {}; ]
							\edge [inv];
							[. \node [inv] (o) {}; ]
						]
						\edge [inv];
						[. \node [inv] {}; ]
					]
					\edge [inv];
					[. \node [bypass] (q)  {};
						\edge [hide];
						[. \node [bypass] {}; ]
						\edge [inv];
						[. \node [inv] {};
							\edge [inv];
							[. \node [inv] {}; ]
							\edge [inv];
							[. \node [label=below:\footnotesize $u$] (u) {}; ]
						]
					]
				]
			]
			\node [above left = 0pt and 5pt of a] {\footnotesize $\T_7=\T_8$};
			\draw (a) edge [hide, bend left = 20] (q);
			\draw (q) edge [hide, bend left = 10] (u);
			\draw (a) edge [bend left = 30] (u);
			\draw (a) edge [bend right = 20] (d);
		\end{tikzpicture}
	}
	\caption{Example for the tree contraction algorithm. Alternatingly, left and right internal leaves are pruned-and-bypassed.
		The contraction patterns introduced in $\T_i$ are hidden in the edges of $\T_{i+1}$, e.g. the pattern $(m,o)$ introduced in $\T_0$
		is hidden in the edge $(l,o)$ in $\T_1$.}
	\label{fig:contract-compress}
\end{figure}

\begin{lemma}
	\label{lem:tree-contract-seq}
	There is an FOM-computable function which maps a binary tree $\T$ and a number $0 \le i \le m$ to $\T_i$. 
\end{lemma}
	
\begin{proof}
	Similar proofs are given in \cite{Buss93,ElberfeldJT12}.
	The main observation is that the least common ancestor of two nodes in $\T_i$
	is the same as their least common ancestor in $\T_0$.
	Therefore it suffices to compute the set of leaves of $\T_i$, which directly also yields the inner nodes
	as the least common ancestors of any two leaves.
			
	First, the number $m$ is FOM-definable using the BIT-predicate ($\lfloor \log n \rfloor + 1$ is the largest
	number $i$ with $\mathrm{BIT}(n,i) = 1$).
	The leaves of the tree $\T_i$ are the left- and rightmost leaf of $\T_0$, together with the internal leaves.
	The internal leaves of a tree $\T_{2i}$ are the internal leaves of $\T_0$ whose number in $\T_0$ is divided by $2^i$.
	The internal leaves of $\T_{2i+1}$ are the internal leaves of $\T_{2i+2}$ and all internal leaves of $\T_{2i}$
	which are right children.
\end{proof}
	
As in the sequential tree contraction algorithm we obtain a hierarchical definition
by taking the set of all contraction patterns which are formed in every prune-and-bypass operation
together with the largest pattern.
We call this hierarchical definition $\mathrm{CP}(\T)$.
Figure~\ref{fig:hier-def} shows the hierarchical definition obtained from
the example in Figure~\ref{fig:contract-compress}. 
The blue pattern is the largest pattern (the subtree pattern $a$).
	
Explicitly written down, we have
\[
	\mathrm{CP}(\T) = \{ a, (e,g), (h,j), (m,o), (s,u), (l,o), (k,q), (c,d), (b,d), (k,u) \}.
\]
	
\begin{proposition}
	\label{prop:hierarchical-def-binary}
	There is an FOM-computable function which maps a binary tree $\T$ to $(\T,\mathrm{CP}(\T))$,
	which is a hierarchical definition of depth $O(\log n)$ and width at most 5.
\end{proposition}
	
\begin{proof}
	The FOM-definition of $\mathrm{CP}(\T)$ follows easily from Lemma~\ref{lem:tree-contract-seq}.
	In every round of the algorithm all new contraction patterns are pairwise disjoint.
	Furthermore every new contraction pattern is maximal, i.e. it is not a subpattern of a previously introduced contraction pattern,
	because every previously introduced contraction patterns is hidden in some edge.
	Hence, the depth of $\mathrm{CP}(\T)$ is bounded by the number of rounds, which is $O(\log n)$.
	It remains to show that $(\T,\mathrm{CP}(\T))$
	has width at most 5.
			
	Every edge $(u,w')$ in a tree $\T_i$ which is not contained in $\T$ originates from an earlier prune-and-bypass operation,
	which implies that $\mathrm{CP}(\T)$ contains the maximal subpattern $p$ hidden in $(u,w')$.
	Let $w$ be the pruned leaf and $v$ be the parent node of $w$. Thus, the three edges involved in the prune-and-bypass operation
	are $(u,v)$, $(v,w)$ and $(v,w')$.
	The direct subpatterns of $p$ must be hidden in the three patterns $(u,v)$, $(v,w)$ and $(v,w')$.
	This proves that $p$ has at most three direct subpatterns (if say $(u,v)$ is an edge of $\T = \T_0$ then there 
	is no contraction pattern yet hidden in $(u,v)$, and similarly for $(v,w)$ and $(v,w')$; 
	hence the number of direct subpatterns of $p$ can be smaller than three).
	Furthermore $p$ has exactly two boundary nodes, namely the leaf $w$ and its parent node $v$.
	The largest pattern has three boundary nodes (the root of the tree and the outermost leaves)
	and at most two direct subpatterns.
	Hence the width of $\mathrm{CP}(\T)$ is bounded by 5.
\end{proof}
	
\subsection{Compression to size $n / \log n$} \label{sec-n/logn}
We improve Proposition~\ref{prop:hierarchical-def-binary}
by constructing a hierarchical definition in which many patterns are equivalent in a strong sense.
This will be crucial for proving the size bound $O(n / \log n)$ for  tree straight-line programs in Section~\ref{sec-small-TSLP}.
The result from this section will be only needed for our applications in 
Section~\ref{sec-regular} (but not Section~\ref{sec-finite-alg}).
		
For a pattern $p \in P$ the set $P[p] = \{ q \in P \mid q \le p \}$
forms a hierarchical definition of $\T[p]$.
Two patterns $p_1,p_2 \in P$ are {\em equivalent}
if the structures $(\T[p_1],P[p_1])$ and $(\T[p_2],P[p_2])$ are isomorphic.
Alternatively, $p_1$ is equivalent to $p_2$ if the subtrees of the pattern tree rooted in $p_1$ and $p_2$ are isomorphic.
The goal is to construct an FOM-definable hierarchical definition
in which there are at most $O(n/\log n)$ inequivalent patterns.
We follow the method of \cite{GMT88}, in which the authors describe a parallel tree contraction algorithm
which uses $O(n/\log n)$ processors on an EREW PRAM. 
The idea is to decompose the input tree into $O(n/\log n)$ many patterns of size $O(\log n)$.
		
We briefly summarize the notions and results from \cite{GMT88}.
Let $\T$ be a binary tree with $n$ nodes and let $1 < m \le n$ be an integer.
An inner node $v$ in $\T$ is {\em $m$-critical} if
$\lceil |\T[v]|/m \rceil \neq \lceil |\T[w]| / m \rceil$
for all children $w$ of $v$, which is equivalent to saying that there exists a multiple $m'$ of $m$
such that $|\T[w]| \le m' < |\T[v]|$ for all children $w$ of $v$.
Consider the set $C$ of all $m$-critical nodes and the subgraph of $\T$ induced by $V(\T) \setminus C$.
Each of its connected components is a tree $\T[p]$ for some pattern $p$
(this is implicitly stated in \cite[Lemma~9.2.1]{GMT88}).
These patterns $p$ are called {\em $m$-bridges}.\footnote{Our definition slightly deviates from the one
	given in \cite{GMT88} where a bridge also contains the neighbouring critical nodes as ``attachments''.}
It was proven in \cite{GMT88} that each $m$-bridge has size at most $m$ 
and that the number of $m$-critical nodes in $\T$ is at most $2n/m-1$.

\begin{proposition} \label{prop-n/log n}
	For every constant $\ell$, 
	there is an FOM-computable function which maps a binary tree $\T$ of size $n$ and with $\ell$ node labels
	to a hierarchical definition $(\T,P)$
	of constant width, depth $O(\log n)$, and with $O(n/\log n)$ inequivalent patterns.
\end{proposition}
	
\begin{proof}
	\begin{figure}
		\centering
		\tikzstyle{crit} = [fill =  red!50, draw = red!50]
		\tikzstyle{b} = [fill = white, draw = black!50]
		\tikzstyle{hide} = [draw = black!30, densely dashed]
		\begin{tikzpicture}[nn, level distance = 12pt, sibling distance = 4pt]
			\Tree
			[. \node [b] (a) {};
				\edge [hide];
				[. \node [crit] (b) {};
					\edge [hide];
					[. \node (c) {};
						\edge [hide];
						[. \node [b] (d) {}; ]
						[. \node (e) {}; ]
					]
					\edge [hide];
					[. \node (f) {};
						\edge [hide];
						[. \node [crit] (g) {};
							\edge [hide];
							[. \node [crit] (h) {};
								\edge [hide];
								[. \node (i) {};
									\edge [hide];
									[. \node [b] (j) {}; ]
									[. \node (k) {}; ]
								]
								\edge [hide];
								[. \node (l) {};
									[. \node (m) {}; ]
									\edge [hide];
									[. \node [b] (n) {}; ]
								]
							]
							\edge [hide];
							[. \node (o) {};
								\edge [hide];
								[. \node [crit] (p) {};
									\edge [hide];
									[. \node (q) {};
										\edge [hide];
										[. \node [b] (r) {}; ]
										[. \node (s) {}; ]
									]
									\edge [hide];
									[. \node (t) {};
										[. \node (u) {}; ]
										\edge [hide];
										[. \node [b] (v) {}; ]
									]
								]
								[. \node (w) {};
									[. \node (x) {}; ]
									[. \node (y) {}; ]
								]
							]
						]
						[. \node (z) {}; ]
					]
				]
				\edge [hide];
				[. \node (aa) {};
					[. \node (bb) {}; ]
					\edge [hide];
					[. \node [b] (cc) {}; ]
				]
			]
		\end{tikzpicture}
		\tikzset{every tree node/.append style={fill=none, draw=none}, edge from parent/.style = {draw=none}}
		\tikzstyle{crit} = [fill = red, draw = red]
		\tikzstyle{b} = [fill = white, draw = black]
		\begin{tikzpicture}[nn, level distance = 12pt, sibling distance = 4pt]
			\Tree
			[. \node [b] (a) {};
				[. \node [crit] (b) {};
					[. \node (c) {};
						[. \node [b] (d) {}; ]
						[. \node (e) {}; ]
					]
					[. \node (f) {};
						[. \node [crit] (g) {};
							[. \node [crit] (h) {};
								[. \node (i) {};
									[. \node [b] (j) {}; ]
									[. \node (k) {}; ]
								]
								[. \node (l) {};
									[. \node (m) {}; ]
									[. \node [b] (n) {}; ]
								]
							]
							[. \node (o) {};
								[. \node [crit] (p) {};
									[. \node (q) {};
										[. \node [b] (rr) {}; ]
										[. \node (s) {}; ]
									]
									[. \node (t) {};
										[. \node (u) {}; ]
										[. \node [b] (v) {}; ]
									]
								]
								[. \node (w) {};
									[. \node (x) {}; ]
									[. \node (y) {}; ]
								]
							]
						]
						[. \node (z) {}; ]
					]
				]
				[. \node (aa) {};
					[. \node (bb) {}; ]
					[. \node [b] (cc) {}; ]
				]
			]
			\draw (a) edge [bend right = 10] (b);
			\draw (a) edge [bend left = 10] (cc);
			\draw (b) edge [bend right = 10] (d);
			\draw (b) edge [bend left = 10] (g);
			\draw (g) edge [bend right = 10] (h);
			\draw (g) edge [bend left = 10] (p);
			\draw (h) edge [bend right = 10] (j);
			\draw (h) edge [bend left = 10] (n);
			\draw (p) edge [bend right = 10] (rr);
			\draw (p) edge [bend left = 10] (v);
		\end{tikzpicture}
		\caption{Removing the $6$-critical nodes (in red) and certain auxiliary nodes (in white), yields a disjoint union of patterns (the set $B$ in the proof of Proposition~\ref{prop-n/log n}),
			which are depicted in black. The binary tree $\T_C$ is obtained by contracting all patterns in $B$.}
		\label{fig:bridges}	
	\end{figure}
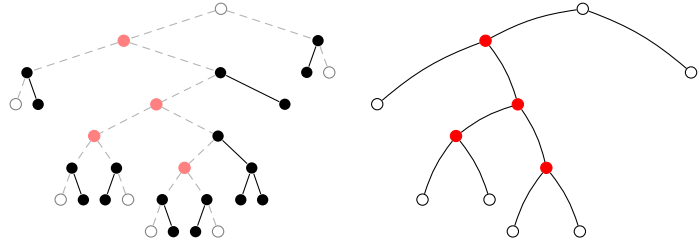
	
	Let $m = \Theta(\log n)$, which will be made explicit in the following. The number $m \leq n$ will be FOM-definable,
	which implies that the set of $m$-critical nodes will be also FOM-definable.
	The idea is to contract all $m$-bridges in $\T$ and apply Proposition~\ref{prop:hierarchical-def-binary}.
	However, the resulting tree is not necessarily a binary tree and may not be rooted in the root of $\T$,
	see Figure~\ref{fig:bridges} for an example.
	Define $C$ to be the set of all $m$-critical nodes together with the root of $\T$.
	Furthermore we add certain leaf nodes to $C$. If the left (resp., right) subtree below a node $v \in C$ contains
	no node in $C$, then we add an arbitrary leaf (e.g. the smallest one with respect to the built-in order on the domain) 
	from the left (resp., right) subtree to $C$.
	Since we add at most two nodes for each $m$-critical node, we still have $|C| = O(n/\log n)$.
	Notice that $V(\T) \setminus C$ is a disjoint union of sets $V[p]$
	for certain patterns $p$ in $\T$.
	Let $B$ be the set of all these patterns, which can be seen to be FOM-definable.
	If we contract all patterns in $B$ to edges we obtain the binary tree $\T_C$ over the node set $C$ of size $O(n/\log n)$.
	Since the set of $m$-critical nodes is FOM-definable, also the set $C$ and the binary tree $\T_C$ are
	FOM-definable.
			
	We can now apply Proposition~\ref{prop:hierarchical-def-binary}
	to obtain a hierarchical definition $\mathrm{CP}(\T_C)$ of depth $O(\log |\T_C|) = O(\log (n / \log n)) = O(\log n)$ and
	width 5.
	Since the size of $\mathrm{CP}(\T_C)$ is at most $O(n/\log n)$, the number of inequivalent patterns is also bounded
	by the same number.
			
	Let us count the number of non-isomorphic trees $\T[p]$ where $p \in B$.
	Let $\ell$ be the number of node labels in $\T$.
	Since every $m$-bridge has size at most $m$, this also holds for all $p \in B$. By inserting a distinguished leaf node to context patterns $p$,
	we can instead count the number of binary trees with at most $m+1$ nodes and $\ell+1$ labels. Using the formula for the Catalan
	number, one can upper-bound the number of such trees by 
	$\frac{4}{3} (4\ell+4)^{m+1} \leq (4\ell+4)^{m+2}$, see e.g. \cite[Lemma~1]{GHJLN17}.
	Hence by choosing $m = \lfloor 1/2 \cdot \log_{4\ell+4}(n) -2\rfloor \in \Theta(\log n)$ (which is indeed FOM-definable)
	($\ell$ is a constant), the number of non-isomorphic trees $\T[p]$ for $p \in B$ is bounded by $\sqrt{n} \in o(n / \log n)$.
			
	For every $p \in B$ we define a canonical well-nested set of patterns $Q_p$ which contains for each covered node $v \in V[p]$
	the maximal subpattern $q \le p$ which is rooted in $v$.
	Clearly, $Q_p$ is a hierarchical definition for $\T[p]$ whose size and depth is bounded by $O(\log n)$.
	Its width is at most 3 because every pattern in $Q_p$ has exactly one boundary node (its root) and at most two direct subpatterns.
	Furthermore, $Q_p$ is FO-definable from $p$ and canonical in the sense that the isomorphism type of the pattern tree of $Q_p$
	is determined by the isomorphism type of $\T[p]$. 
	Hence the number of inequivalent patterns in 
	$Q = \bigcup_{p \in B} Q_p$ is bounded by the number of patterns of size at most $m$, which by the above calculation
	is bounded by $o(n / \log n)$.
			
	Now we claim that $P = \mathrm{CP}(\T_C) \cup Q$ is a hierarchical definition for $\T$ with the desired properties.
	Clearly the largest pattern is contained in $P$, and both $\mathrm{CP}(\T_C)$ and $Q$ are well-nested.
	Furthermore, since each pattern $p \in B$ is hidden in some edge of $\T_C$, also $P$ is well-nested.
	The depth of $P$ is bounded by $O(\log n)$ and
	the number of inequivalent patterns is $O(n/\log n) + o(n / \log n) = O(n/\log n)$.
	To prove that the width of $P$ is bounded by some constant, we notice that the patterns $p \in B$
	are the maximal subpatterns of $\T$ hidden in some edge of $\T_C$.
	More precisely, if $p \in B$ is a direct subpattern of a pattern $q \in \mathrm{CP}(\T_C)$
	then $p$ must be hidden in an edge of $\T_C[q]$ which is not covered by any subpattern $q'$ of $q$.
	Since the branching tree of $q$ has size at most 5, there are at most 4 such possible edges.
	This proves that the width of $P$ is at most 9.
\end{proof}

\subsection{Non-binary trees}
	        
Now we extend Proposition~\ref{prop-n/log n} to arbitrary $k$-ordered trees, for any constant $k \geq 1$.
	
\begin{proposition}
	\label{prop:hierarchical-def-arb}
	For all constants $k, \ell \geq 1$,
	there is an FOM-computable function which maps a $k$-ordered tree $\T$ of size $n$ and with 
	$\ell$ node labels to a hierarchical definition $(\T,P)$
	of depth $O(\log n)$, constant width, and with $O(n/\log n)$ inequivalent patterns.
\end{proposition}

\begin{proof}		
	\begin{figure}
				
		\centering
				
		\begin{tikzpicture}[inner sep = 0]
			\coordinate (0);
			\coordinate [below left = 30pt and 60pt of 0] (1);
			\coordinate [right = 10pt of 1] (2) {};
			\coordinate [right = 50pt of 2] (3) {};
			\coordinate [right = 50pt of 3] (4) {};
			\coordinate [right = 10pt of 4] (5);
			\node [draw, fill = black, circle, minimum size = 3pt, above right = 8pt and 25pt of 3] (6) {};
					
			\coordinate [below left = 15pt and 5pt of 2] (7);
			\coordinate [right = 10pt of 7] (8);
					
			\coordinate [below left = 15pt and 30pt of 3] (9) ;
			\coordinate [right = 15pt of 9] (10);
			\coordinate [right = 3pt of 10] (10a);
			\coordinate [right = 30pt of 10a] (11);
			\coordinate [right = 3pt of 11] (11a);
			\coordinate [right = 5pt of 11a] (12);
					
			\coordinate [below left = 15pt and 20pt of 10] (29);
			\coordinate [right = 30pt of 29] (30);
			\coordinate [right = 3pt of 30] (30a);
			\coordinate [right = 10pt of 30] (31);
					
			\coordinate [below left = 15pt and 25pt of 30] (32);
			\coordinate [right = 15pt of 32] (33);
			\coordinate [right = 30pt of 33] (34);
					
			\coordinate [below left = 15pt and 5pt of 4] (13);
			\coordinate [right = 10pt of 13] (14);
					
			\coordinate [below left = 25pt and 30pt of 33] (15);
			\coordinate [right = 10pt of 15] (16);
			\coordinate [right = 15pt of 16] (17);
			\coordinate [right = 40pt of 17] (18);
			\node [draw, fill = black, circle, minimum size = 3pt, above right = 5pt and 5pt of 17] (19) {};
					
			\coordinate [below left = 15pt and 5pt of 16] (20);
			\coordinate [right = 10pt of 20] (21);
			\coordinate [right = 5pt of 21] (22);
			\coordinate [right = 10pt of 22] (23);
					
			\coordinate [below left = 15pt and 5pt of 11] (24);
			\coordinate [right = 13pt of 24] (25);
					
			\coordinate [right = 15pt of 17] (26);
			\coordinate [below left = 15pt and 5pt of 26] (27);
			\coordinate [right = 10pt of 27] (28);
					
			\coordinate [right = 15pt of 26] (35);
			\coordinate [below left = 15pt and 5pt of 35] (36);
			\coordinate [right = 10pt of 36] (37);
					
			\foreach \x/\y in {0/1,1/5,5/0,
				6/3,6/4,
				2/7,2/8,3/9,9/10,10a/11,11a/12,3/12,4/13,4/14,
				33/15,15/16,16/17,17/18,33/18, 16/20, 16/21, 17/22, 17/23, 11/24, 11a/25,
				19/26, 26/27, 26/28, 19/17} {
				\draw (\x.center) edge (\y.center);
			}
					
			\tikzstyle{pat} = [fill = col1!15]
					
			\begin{scope}[on background layer]
				\path [pat] (6.center) -- (3) -- (4) -- cycle;
				\path [pat] (33) -- (15) -- (17) -- (19.center)-- (26) -- (18) -- (33) -- cycle;
				\path [pat] (3) -- (9) -- (12) -- cycle;
				\path [pat] (10) -- (10a) -- (31) -- (29) -- cycle;
				\path [pat] (30) -- (30a) -- (34) -- (32) -- cycle;
				\path [pat] (13) -- (4) -- (14);
				\path [pat] (24) -- (11) -- (11a) -- (25);
				\path [pat] (20) -- (16) -- (21);
				\path [pat] (35) -- (36) -- (37);
			\end{scope}
					
			\foreach \x/\y in {3/9,3/12,9/10, 10a/11, 11a/12,13/4,4/14,24/11,11a/25,20/16,16/21, 10/29, 29/30,
				30a/31, 31/10a, 30/32, 32/33, 33/34, 30a/34, 35/36, 35/37} {
				\draw [line width = 1pt, line cap = round, col1] (\x) -- (\y);
			}
		\end{tikzpicture}
				
		\caption{The shaded pattern contains four maximal subpatterns (framed in blue) which are unions of zones.}
		\label{fig:zones}
				
	\end{figure}
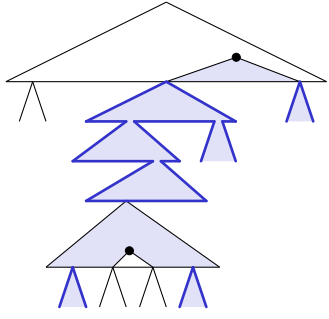
	The idea is that one can embed $\T$ into an FOM-definable binary tree $\T'$
	and transform a hierarchical definition for $\T'$ into one for $\T$.
	More precisely, an {\em embedding} of $\T$ into $\T'$ is an injective function $\varphi: V(\T) \to V(\T')$
	such that $\varphi$ maps the root of $\T$ to the root of $\T'$,
	$u \preceq_T v$ if and only if $\varphi(u) \preceq_{T'} \varphi(v)$
	and $\varphi$ preserves the depth-first order of nodes in $\T$, see \cite{ElberfeldJT12}.
	Nodes which are not in the image of $\varphi$ are labelled by a fresh symbol to distinguish them from nodes in the image of $\varphi$.
	Each node $v \in V(\T)$ defines the {\em zone}
	\[
		V[\varphi(u)] \setminus \bigcup_{w \text{ child of } v} V[\varphi(w)].
	\]
	The set of all zones form a partition of $V(\T')$.
	We require that the size of each zone is bounded in a function of the maximum out-degree $k$,
	which can be done by embedding a node with $r$ children into a chain of at least $r-1$ binary nodes.
			
	From Proposition~\ref{prop:hierarchical-def-binary} we obtain a hierarchical definition $P' = \mathrm{CP}(\T')$ for $\T'$
	of depth $O(\log n)$ and width at most 5, and by Proposition~\ref{prop-n/log n} it has only $O(n / \log n)$ inequivalent patterns.
	Notice that the patterns in $P'$ can intersect arbitrarily with the zones,
	as illustrated in Figure~\ref{fig:zones}.
	We adapt $P'$ in such a way that every pattern is a union of zones.
	Since $\varphi$ respects the ancestor relation and the depth-first order, this directly yields a hierarchical definition of $\T$
	(by taking the preimages under $\varphi$).
			
	For a pattern $p' \in P'$ let $Z(p') \subseteq V[p']$ be the union of all zones which are contained in $V[p']$.
	Notice that $V[p'] \setminus Z(p')$ has constant size (it is contained in at most two zones)
	and that $Z(p')$ can be written (uniquely) as a disjoint union of a constant number of maximal subpatterns of $p'$
	(the constants only depend on the maximum out-degree of $\T$ and can be set to $1 + 2 (k-1)$).
	We denote the set of these subpatterns by $S(p')$, which are framed in blue in Figure~\ref{fig:zones}.
	Define the set $P = \bigcup_{p' \in P'} S(p')$, which is clearly FO-definable from $P'$.
	We claim that (1) $P$ is a hierarchical definition for $\T'$ of depth $O(\log n)$,
	(2) its width is bounded by some constant, and (3) the number of inequivalent patterns in $P$ is $O(n/\log n)$.
			
	Clearly the largest pattern is contained in $P$ and we need to verify that $P$ is well-nested.
	Observe that for all $p',q' \in P'$ we have:
	\begin{itemize}
		\item if $p' \le q'$ then $Z(p') \subseteq Z(q')$, and
		\item if $V[p'] \cap V[q'] = \emptyset$ 
		      then $Z(p') \cap Z(q') = \emptyset$.
	\end{itemize}
	Consider $p \in S(p')$ and $q \in S(q')$.
	If $p'$ and $q'$ are disjoint, then also $p$ and $q$ are disjoint.
	If $p' \le q'$ then $p$ is a subpattern of some pattern $r \in S(q')$.
	If $r = q$ then $p \le q$, otherwise $q$ is disjoint from $r$ and therefore also from $p$.
	This concludes the proof that $P$ is well-nested.
	Furthermore these observations imply that, if $p \in S(p')$ and $q \in S(q')$ such that $p < q$,
	then $p' < q'$.
	Hence the depth of $P$ is bounded by the depth of $P'$.
		
	Next we show that the width of $P$ is bounded by some constant.
	Consider a pattern $p \in P$ and let $p' \in P'$ be the minimal pattern such that $p \in S(p')$.
	Let $p_1', \dots, p_m' \in P'$ be the direct subpatterns of $p'$. Since the width of $P'$ 
	is constant, $m$ is bounded by a constant. Moreover, also the size of the boundary 
	$\partial p'$ is bounded by a constant. 
	Since $\partial p' = V[p'] \setminus \bigcup_{i=1}^m V[p_i']$ and $Z(p') \setminus \bigcup_{i=1}^m Z(p_i')$
	only differ by a constant number of nodes, it follows that the size of 
	$Z(p') \setminus \bigcup_{i=1}^m Z(p_i')$ is bounded by a constant.
	Note that $V[p] \subseteq Z(p')$. Moreover, by the minimality of $p'$, every pattern in 
	$\bigcup_{i=1}^m S(p_i')$ is either disjoint or properly contained in $V[p]$. 
	It follows that 
	the boundary $\partial p$ is contained in $V[p] \setminus \bigcup_{i=1}^m Z(p_i') \subseteq Z(p') \setminus \bigcup_{i=1}^m Z(p_i')$.
	Hence, there is a constant that bounds the size of every boundary $\partial p$ for $p \in P$.
			
	To show that $P$ has bounded width, it 
	remains to show that the number of direct subpatterns of a pattern $p \in P$ is bounded by 
	a constant.
	Consider such a direct subpattern $q \in P$, i.e. $q \lessdot p$.
	Choose $q' \in P'$ maximal such that $q \in S(q')$ and choose $p' \in P'$ minimal such that $p \in S(p')$.
	We already know that $q' < p'$ and we claim that in fact $q' \lessdot p'$ holds.
	Towards a contradiction let $r' \in P'$ with $q' < r' < p'$.
	By the choice of $q'$ and $p'$ we have $p,q \notin S(r')$.
	This means that $q$ is a proper subpattern of some pattern $r \in S(r')$.
	Since both $r$ and $p$ share the subpattern $q$, the patterns $r$ and $p$ are comparable.
	Furthermore, since $r$ is a subpattern of some pattern in $S(p')$,
	we must have $r \le p$.
	We conclude that $q \le r \le p$ and $q \neq r \neq p$, which contradicts $q \lessdot p$.
	This proves that the number of direct subpatterns of $p$ is bounded by $|S(p')|$ (a constant)
	times the number of direct subpatterns of $p'$. The latter is bounded by the width of $P'$, which is a constant.
			
	It remains to show that $P$ has $O(n/\log n)$ inequivalent patterns.
	First, if $p',q' \in P'$ are equivalent then there is an isomorphism $\psi : (\T'[p'], P'[p']) \to (\T'[q'], P'[q'])$.
	This isomorphism induces a bijection between $S(p')$ and $S(q')$
	which maps each pattern $p \in S(p')$ to the isomorphic pattern $\psi(p) \in S(q')$. 
	Since each set $S(p')$ has constant size, it suffices to show that $\psi$ yields an isomorphism
	from $(\T'[p], P[p])$ to $(\T'[\psi(p)], P[\psi(p)])$ for all $p \in S(p')$. 
	Recall that if  $r < p$ for a pattern $r \in S(r')$,  then $r' < p'$ and thus $r' \in P'[p']$, and similarly
	for $\psi(p)$.  It follows that the isomorphism types of $(\T'[p], P[p])$ and $(\T'[\psi(p)], P[\psi(p)])$
	are completely determined by the isomorphism types of $(\T'[p'], P'[p'])$ and $(\T'[q'], P'[q'])$, which are equal.
\end{proof}

\subsection{Normal form}
	
In a final step, we bring the computed hierarchical definition into a normal form.
A hierarchical definition $(\T,P)$ is in {\em normal form} if for each pattern $p \in P$ one of the following
two cases holds:
	
\begin{enumerate}
	\item If $u$ is the root of $p$, then $\partial p = \{u\}$ and every direct subpattern of $p$
	      is a subtree pattern (which must be rooted in a child of $u$). \label{type:1}
	\item $p$ has exactly two direct subpatterns $p_1$ and $p_2$, and 
	      		
	      $V[p]$ is the disjoint union of $V[p_1]$ and $V[p_2]$. \label{type:2}
	      		
\end{enumerate}
Figure~\ref{fig:normal-form} illustrates the four types that a pattern can be decomposed
where the first two patterns have type \ref{type:1} and the latter two have type \ref{type:2}.
\begin{figure}
	\begin{tikzpicture}[inner sep = 0, every node/.style={circle, minimum size = 3pt, fill = black}]
		\node (0) {};
		\node [below left = 15pt and 35pt of 0, fill = col1] (1) {};
		\node [right = 20pt of 1, fill = col1] (2) {};
		\node [right = 20pt of 2, fill = col1] (3) {};
		\node [right = 20pt of 3, fill = col1] (4) {};
				
		\coordinate [below left = 20pt and 8pt of 1] (1a);
		\coordinate [below right = 20pt and 8pt of 1] (1b);
		\coordinate [below left = 20pt and 8pt of 2] (2a);
		\coordinate [below right = 20pt and 8pt of 2] (2b);
		\coordinate [below left = 20pt and 8pt of 3] (3a);
		\coordinate [below right = 20pt and 8pt of 3] (3b);
		\coordinate [below left = 20pt and 8pt of 4] (4a);
		\coordinate [below right = 20pt and 8pt of 4] (4b);
				
		\tikzstyle{pat} = [fill = col1!15]
				
		\begin{scope}[on background layer]
			\path [pat] (1.center) -- (1a) -- (1b) -- cycle;
			\path [pat] (2.center) -- (2a) -- (2b) -- cycle;
			\path [pat] (3.center) -- (3a) -- (3b) -- cycle;
			\path [pat] (4.center) -- (4a) -- (4b) -- cycle;
		\end{scope}
				
		\foreach \x/\y in {0/1,0/2,0/3,0/4} {
			\draw [line width = 1pt, line cap = round] (\x) -- (\y);
		}
		\foreach \x/\y in {1/1a,1/1b,2/2a,2/2b,3/3a,3/3b,4/4a,4/4b} {
			\draw [line width = 1pt, line cap = round, col1] (\x) -- (\y);
		}
	\end{tikzpicture}
	\hfill
	\begin{tikzpicture}[inner sep = 0, every node/.style={circle, minimum size = 3pt, fill = black}]
		\node (0) {};
		\node [below left = 15pt and 35pt of 0, fill = col1] (1) {};
		\node [right = 20pt of 1, fill = col1] (2) {};
		\node [right = 20pt of 2, fill = black!30] (3) {};
		\node [right = 20pt of 3, fill = col1] (4) {};
				
		\coordinate [below left = 20pt and 8pt of 1] (1a);
		\coordinate [below right = 20pt and 8pt of 1] (1b);
		\coordinate [below left = 20pt and 8pt of 2] (2a);
		\coordinate [below right = 20pt and 8pt of 2] (2b);
		\coordinate [below left = 20pt and 8pt of 3] (3a);
		\coordinate [below right = 20pt and 8pt of 3] (3b);
		\coordinate [below left = 20pt and 8pt of 4] (4a);
		\coordinate [below right = 20pt and 8pt of 4] (4b);
				
		\tikzstyle{pat} = [fill = col1!15]
				
		\begin{scope}[on background layer]
			\path [pat] (1.center) -- (1a) -- (1b) -- cycle;
			\path [pat] (2.center) -- (2a) -- (2b) -- cycle;
			\path [pat] (4.center) -- (4a) -- (4b) -- cycle;
		\end{scope}
				
		\foreach \x/\y in {0/1,0/2,0/4} {
			\draw [line width = 1pt, line cap = round] (\x) -- (\y);
		}
		\draw [line width = 1pt, draw=black!30, line cap = round] (0) -- (3);
		\foreach \x/\y in {1/1a,1/1b,2/2a,2/2b,4/4a,4/4b} {
			\draw [line width = 1pt, line cap = round, col1] (\x) -- (\y);
		}
		\foreach \x/\y in {3/3a,3/3b} {
			\draw [line width = 1pt, draw=black!30, line cap = round] (\x) -- (\y);
		}
	\end{tikzpicture}
	\hfill
	\begin{tikzpicture}[inner sep = 0, every node/.style={circle, minimum size = 3pt, fill = col1}]
		\node (1) {};
		\coordinate [below left = 30pt and 15pt of 1] (1a);
		\coordinate [below right = 30pt and 15pt of 1] (1b);
		\node [below = 20pt of 1] (2) {};
		\coordinate [below left = 25pt and 15pt of 2] (2a);
		\coordinate [below right = 25pt and 15pt of 2] (2b);
				
		\tikzstyle{pat} = [fill = col1!15]
				
		\begin{scope}[on background layer]
			\path [pat] (1.center) -- (1a) -- (1b) -- cycle;
			\path [pat] (2.center) -- (2a) -- (2b) -- cycle;
		\end{scope}
		
		\foreach \x/\y in {1/1a,1/1b,2/2a,2/2b} {
			\draw [line width = 1pt, line cap = round, col1] (\x) -- (\y);
		}
				
	\end{tikzpicture}
	\hfill
	\begin{tikzpicture}[inner sep = 0, every node/.style={circle, minimum size = 3pt, fill = col1}]
		\node (1) {};
		\coordinate [below left = 30pt and 15pt of 1] (1a);
		\coordinate [below right = 30pt and 15pt of 1] (1b);
		\node [below = 20pt of 1] (2) {};
		\coordinate [below left = 25pt and 15pt of 2] (2a);
		\coordinate [below right = 25pt and 15pt of 2] (2b);
		\node [below = 15pt of 2] (3) {};
		\coordinate [right = 12pt of 2a] (3a);
		\coordinate [left = 12pt of 2b] (3b);

		\tikzstyle{pat} = [fill = col1!15]
				
		\begin{scope}[on background layer]
			\path [pat] (1.center) -- (1a) -- (1b) -- cycle;
			\path [pat] (2.center) -- (2a) -- (3a) -- (3.center) -- cycle;
			\path [pat] (2.center) -- (2b) -- (3b) -- (3.center) -- cycle;
		\end{scope}
				
		\foreach \x/\y in {1/1a,1/1b,2/2a,2/2b,3/3a,3/3b} {
			\draw [line width = 1pt, line cap = round, col1] (\x) -- (\y);
		}
				
	\end{tikzpicture}
	\caption{Four types in which a pattern can be decomposed in a normal form hierarchical definition.}
	\label{fig:normal-form}
			
\end{figure}
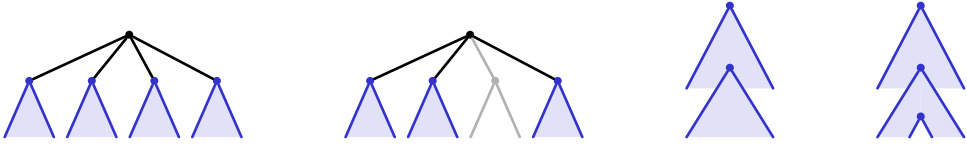
	
\begin{theorem}
	\label{thm:hierarchical-def}
	There is an FOM-computable function which maps a $k$-ordered tree $\T$ of size $n$ to a hierarchical definition $(\T,P)$
	in normal form
	where $P$ has depth $O(\log n)$ and $O(n/\log n)$ inequivalent patterns.
\end{theorem}
	
\begin{proof}
	Let $P$ be the hierarchical definition from Proposition~\ref{prop:hierarchical-def-arb}. Note that $P$ has constant width.
	For each pattern $p \in P$ we introduce new subpatterns corresponding to the subtrees of its branching tree:
	For each node $w \in \partial p$ on the boundary,
	we add the maximal subpattern of $p$ rooted in $w$.
	Furthermore, for each direct subpattern $q \lessdot p$ we add
	the maximal subpattern of $p$ rooted in the root of $q$.
	This ensures that the branching tree of each pattern has height at most 1.
	The depth of $P$ increases at most by a factor of its width.
	Now consider a context pattern $(v,w) \in P$ which has a direct subpattern $(v',w) \in P$ such that $v'$ is a child of $v$.
	To establish normal form it suffices to introduce the pattern $(v,v')$.
	In this step the depth of $P$ increases at most by a factor of two.
			
	Both steps are FO-computable.
	Finally, notice that for any two equivalent patterns we introduce equivalent new subpatterns,
	therefore, the number of inequivalent patterns increases by a constant factor.
\end{proof}

\section{Balancing over free term algebras and arbitrary algebras} \label{sec-small-TSLP}

In this section, we transform the hierarchical decomposition constructed in the previous 
section into a so called tree straight-line programs, or {\em TSLPs} for short.
TSLPs are used as a compressed representation of trees, see \cite{Lohrey15dlt} for a survey. 
Formally, a {\em TSLP} $\G = (N,\Sigma,S,P)$ consists of two disjoint ranked alphabets $N$ and $\Sigma$,
where symbols in $N$ are called {\em nonterminals} and have rank at most one,
a start nonterminal $S \in N$ of rank zero,
and a set of productions $P$.
For each nonterminal $A \in N$ there exists exactly one production $(A \to t) \in P$, 
where $t \in T(\Sigma \cup N)$ if $A$ has rank zero and $t \in C(\Sigma \cup N)$ if $A$ has rank one.
Furthermore the relation $\{(A,B) \in N \times N \mid B \text{ occurs in } t \text{ where } (A \to t) \in P \}$
must be acyclic. These properties ensure that the start nonterminal $S$ derives exactly one term $t \in T(\Sigma)$ by applying
the productions in any order, starting with $S$, as long as possible, see \cite{Lohrey15dlt} for more details (we will give an alternative circuit based
definition below).
A TSLP is in {\em normal form} if  every production has one of the following forms:
\begin{itemize}
	\item $A \to f(A_1, \dots, A_r)$
	\item $A(x) \to f(A_1, \dots, A_{i-1}, x, A_{i+1}, \dots, A_r)$
	\item $A \to B(C)$
	\item $A(x) \to B(C(x))$
\end{itemize}	
We will work here with an alternative definition of TSLPs as circuits over an extended term algebra.
\begin{definition}

	Let $\Sigma$ be a ranked alphabet. The two-sorted algebra $\A(\Sigma)$ consists of the two sorts $T(\Sigma)$ (all terms) and
	$C(\Sigma)$ (all contexts) and the following operations:
	\begin{itemize}
		\item for all $f \in \Sigma$ of rank $r \geq 0$ the operation $f: T(\Sigma)^r \to T(\Sigma)$ that maps  
		      $(t_1, \dots, t_r)$ to  $f(t_1, \dots, t_r)$,
		\item for all $f \in \Sigma$ of rank $r \geq 1$ and $1 \le i \le r$ the operation $\hat f_i: T(\Sigma)^{r-1} \to C(\Sigma)$
		      with $\hat f_i(t_1, \dots, t_{r-1}) = f(t_1, \ldots, t_{i-1}, x, t_{i}, \ldots, t_{r-1})$,
		\item the substitution operation
		      $\mathrm{sub}: C(\Sigma) \times T(\Sigma) \to T(\Sigma)$ with
		      $\mathrm{sub}(s,t) = s(t)$,
		\item the composition operation $\circ: C(\Sigma) \times C(\Sigma) \to C(\Sigma)$ with
		      $\circ(s,t) = s(t)$.
	\end{itemize}

\end{definition}
	
\begin{theorem}[universal balancing theorem]
	\label{thm:main-tslp}
	Given a term $t$ over $\Sigma$ of size $n$,
	one can compute in $\TC^0$ a normal form TSLP for $t$ of size $O(n/\log n)$ and depth $O(\log n)$.
	The TSLP is given as a circuit over $\A(\Sigma)$ in EC-representation.
\end{theorem}

\begin{proof}
	In $\TC^0$ we convert $t$ into a $\Sigma$-labelled tree $\T$ in ancestor representation (see Theorem~\ref{lem:term-anc}) and apply Theorem~\ref{thm:hierarchical-def}
	to obtain a hierarchical definition $(\T, P)$, which has depth $O(\log n)$ and $O(n/\log n)$ many inequivalent patterns.
	We can translate $(\T,P)$ directly into a tree $\T'$ in ancestor representation over the two-sorted algebra $\A(\Sigma)$ that evaluates to $t$:
	Patterns of rank zero (resp., one) are nodes of sort $T(\Sigma)$ (resp., $C(\Sigma)$), and the children of a pattern are its 
	direct subpatterns.
	The pattern type (whether the pattern is a subtree pattern or a context pattern) determines the operator of the corresponding node.
	The ancestor relation is FO-definable since pattern $p$ is an ancestor of pattern $q$
	if and only if $\T[q] \subseteq \T[p]$.
	From the tree $\T'$ we can compute in $\TC^0$
	by Lemma~\ref{lem:min-dag} the EC-representation of the minimal dag $\C$, which is a circuit over the structure $\A(\Sigma)$.
	Since the number of inequivalent patterns of the hierarchical definition is $O(n/\log n)$ and its depth is $O(\log n)$,
	the circuit $\C$ has size $O(n/\log n)$ and depth $O(\log n)$.
\end{proof}

\begin{definition}
	
	For an algebra $\A$ over $\Sigma$, the two-sorted algebra $\mathcal{F}(\A)$ extends $\A$
	by a second sort $\A[x]$ containing all linear term functions $p: A \to A$.
	The operations of $\mathcal{F}(\A)$ are the following:
	\begin{itemize}
		\item for all $f \in \Sigma$ of rank $r \ge 0$ the operation $f^{\A} : A^r \to A$,
		\item for all $f \in \Sigma$ of rank $r \geq 1$ and and $1 \le i \le r$ the operation $\hat f_i : A^{r-1} \to \A[x]$
		      that maps $(a_1, \dots, a_{r-1}) \in A^{r-1}$ to the linear term function $f^{\A}(a_1, \ldots, a_{i-1}, x, a_{i}, \ldots, a_{r-1})$,
		\item the substitution operation
		      $\mathrm{sub}: \A[x] \times  A \to A$ with
		      $\mathrm{sub}(p,a) = p(b)$,
		\item the composition operation $\circ: \A[x] \times \A[x] \to \A[x]$ that maps $(p,q)$ to the composition
		      of the mappings $p$ and $q$.
	\end{itemize}
\end{definition}

From Theorem~\ref{thm:main-tslp} we immediately get:

\begin{theorem}
	\label{thm:main-algebra}
	Given a term $t$ over $\A$ of size $n$,
	one can compute in $\TC^0$ an equivalent circuit over $\A[x]$ in EC-representation of size $O(n/\log n)$ and depth $O(\log n)$.
\end{theorem}
	
By applying Lemma~\ref{lem:circuit-unfold} we obtain:

\begin{theorem}
	\label{thm:main-algebra-term}
	Given a term over $\A$ of size $n$,
	one can compute in $\TC^0$ an equivalent term over $\A[x]$ of depth $O(\log n)$.
\end{theorem}

\section{Applications}
	
\subsection{Alternative proof of a result by Krebs, Limaye and Ludwig}
	
Recently, Krebs, Limaye and Ludwig presented a similar result to ours \cite{KrebsLL17}.
We will state their result and reprove it using our balancing theorem.
For an $S$-sorted algebra $\A$ we define the algebra $(\B,\A)$ which extends $\A$ by the Boolean sort $\B=\{0,1\}$.
All operations from $\A$ are inherited to $(\B,\A)$, with the addition of the Boolean disjunction $\vee$, conjunction $\wedge$
and negation $\neg$, and for each sort $s \in S$ a multiplexer function
\[
	\mathrm{mp}_s : \{0,1\} \times A_s \times A_s \to A_s
\]
where $\mathrm{mp}_s(b,d_0,d_1) = d_b$.
A circuit family $(\C_n)_{n \ge 0}$ of circuits over $(\B,\A)$ where $\C_n$ has Boolean input gates $x_1, \dots, x_n$ computes a function $f: \{0,1\}^* \to A$. The class $\A$-$\NC^1$ denotes the class of functions computed by a DLOGTIME-uniform circuit family over $(\B,\A)$
of constant fan-in, polynomial size and logarithmic depth.
	
\begin{theorem} \label{thm-kll}
	For every algebra $\A$ there exists a DLOGTIME-uniform $\mathcal{F}(\A)$-$\NC^1$ circuit family
	which computes the value of a given expression over $\A$.
\end{theorem}
	
\begin{proof}
	For a given input expression one can compute in $\TC^0 \subseteq \NC^1$ an equivalent logarithmic depth
	expression $t$ over $\mathcal{F}(\A)$ by Theorem~\ref{thm:main-algebra-term}.
	It suffices to construct a DLOGTIME-uniform circuit family which evaluates $t$.
	Let $n = |t|$ and assume that the depth of $t$ is at most $d = O(\log n)$.
	Furthermore, let $k$ be the maximal arity of an operation in $\mathcal{F}(\A)$ (this is  a constant).
	We can test in $\TC^0$ whether a given string $\rho \in \{1, \dots, k\}^*$ of length at most $d$
	is a valid address string of a path from the root of $t$
	to some node and, if so, we can compute the node label in $\TC^0$.
	Consider the circuit with gates of the form $v_\rho$, where $\rho \in \{1, \dots, k\}^*$
	is a string of length at most $d$, and $v_\rho$
	computes the value of the addressed node,
	or computes some arbitrary value if $\rho$ is not a valid address string.
	With the help of multiplexer gates and the node label information we can clearly compute $v_\rho$ from the gates $v_{\rho \cdot i}$.
	Clearly, the described circuit has depth $O(\log n)$ and the constructed circuit family can be seen to be DLOGTIME-uniform.
\end{proof}
As shown in \cite{KrebsLL17}, many known results on the complexity of expression evaluation problems can be derived
from Theorem~\ref{thm-kll}. The following list is not exhaustive:
\begin{itemize}
	\item Buss' theorem \cite{Bus87}: The expression evaluation problem for $(\{0,1\}, \wedge, \vee, \neg, 0,1)$ belongs to 
	      DLOGTIME-uniform $\NC^1$. 
	\item More generally, for every fixed finite algebra $\A$ the expression evaluation problem belongs to
	      DLOGTIME-uniform $\NC^1$ \cite{Loh01rta}.
	\item Expression evaluation for the semirings $(\N, +, \cdot, 0,1)$ (resp., $(\mathbb{Z}, +, \cdot, 0,1)$)  belongs to $\#\NC^1$
	      (resp., $\mathrm{GapNC}^1$) \cite{BCGR92}.
\end{itemize}    
 
\subsection{Regular expressions and semirings}

\label{sec-regular}
	
It has been shown in \cite[Theorem~6]{GruberH08} that from a given regular expression of size $n$ one can obtain 
an equivalent regular expression of star height $O(\log n)$. Here, we strengthen this result in several directions: (i) the 
resulting regular expression (viewed as a tree) has depth $O(\log n)$ (and not only star height $O(\log n)$), (ii) it can be 
represented by a circuit with only $O(n / \log n)$ nodes, and the construction can be carried out in $\TC^0$ (or, alternatively
in linear time if we use \cite{GHJLN17}).
	
For a finite alphabet $\Sigma$, let $\mathrm{Reg}(\Sigma)$ be the set of regular languages over $\Sigma$.
It forms an algebra with the constants $a \in \Sigma$, $\emptyset$ and $\{ \varepsilon \}$,
the unary operator $^*$ and the binary operations union $+$ and concatenation $\cdot$. It is also known as the free
Kleene algebra. 
	
\begin{theorem} \label{theorem-reg-exp}
	Given a regular expression $t$ over $\Sigma$, one can compute in $\TC^0$ an equivalent circuit over $\mathrm{Reg}(\Sigma)$ in EC-representation
	of size $O(n/\log n)$ and depth $O(\log n)$.
\end{theorem}
		
\begin{proof}
	Let $\mathcal{R} = \mathrm{Reg}(\Sigma)$.
	We claim the following for any 
	context $t \in C(\mathcal{R})$.
	If the parameter $x$ is not below any $*$-operator,
	then the linear term function $t^{\mathcal{R}}$ has the form
	\begin{equation} 
		\label{eqn:kleene-3}
		t^{\mathcal{R}}(x) = axb+c \text{ for some } a,b,c \in {\mathcal{R}},
	\end{equation}
	otherwise it has the form
	\begin{equation}
		\label{eqn:kleene-4}
		t^{\mathcal{R}}(x) = \alpha(axb+c)^* \gamma + \delta \text{ for some } a,b,c, \alpha, \gamma, \delta \in {\mathcal{R}}.
	\end{equation}
	The linear term functions  \eqref{eqn:kleene-3} and \eqref{eqn:kleene-4} are closed under union and left/right concatenation with 
	constants from $\mathcal{R}$. If a term  function is of the form  \eqref{eqn:kleene-3} then its star is of the form 
	\eqref{eqn:kleene-4}.
	The only non-trivial part is to prove that the set of term functions of type \eqref{eqn:kleene-4} are closed under $*$ (this shows then also
	closure under composition).
	Let $\beta = axb+c$ and $p(x) = \alpha \beta^* \gamma + \delta$.
	We claim that
	\[
		p(x)^* = (\alpha \beta^* \gamma + \delta)^* \stackrel{\text{$(\dagger)$}}{=} \delta^* \alpha (\beta + \gamma \delta^* \alpha)^* \gamma \delta^*  + \delta^* 
		= \delta^* \alpha (axb + c + \gamma \delta^* \alpha)^* \gamma \delta^* +  \delta^* .
	\]
	Note that this expression is  indeed  of the form \eqref{eqn:kleene-4}.
	To verify the identity $(\dagger)$,
	one should consider $\alpha, \beta, \gamma$ and $\delta$ as letters. The inclusion
	$\delta^* \alpha (\beta + \gamma \delta^* \alpha)^* \gamma \delta^*  + \delta^*  \subseteq (\alpha \beta^* \gamma + \delta)^*$ 
	is obvious. For the other inclusion, one considers a word $w \in (\alpha \beta^* \gamma + \delta)^*$.  We show that
	$w \in \delta^* \alpha (\beta + \gamma \delta^* \alpha)^* \gamma \delta^*  + \delta^*$. The case that $w \in \delta^*$ is clear.
	Otherwise, $w$ contains at least one occurrence of $\alpha$ and $\gamma$, and we can factorize $w$ uniquely as
	$w = w_0 \alpha w_1 \gamma w_2$, where $w_0, w_2 \in \delta^*$. Moreover, we must have $w_1 \in (\beta + \gamma \delta^* \alpha)^*$,
	which shows that $w \in  \delta^* \alpha (\beta + \gamma \delta^* \alpha)^* \gamma \delta^*$.
	
	Note that a term function of type \eqref{eqn:kleene-3} (resp., \eqref{eqn:kleene-4}) can be represented by the three (resp., six) elements
	$a,b,c \in \mathcal{R}$ (resp., $a,b,c, \alpha,\gamma, \delta \in \mathcal{R}$).
					
	By Theorem~\ref{thm:main-algebra} we compute from the input regular expression $t$ in $\TC^0$ an equivalent circuit $\C$ over
	$\mathcal{R}[x]$ in EC-representation.
	We partition $V(\C)$ into three sets: the set $V_0$ of nodes which evaluate to elements in $\mathcal{R}$, 
	the set $V_1$ of nodes that evaluate 
	to a linear term function of type~\eqref{eqn:kleene-3},
	and the set of nodes that evaluate 
	to a linear term function of type~\eqref{eqn:kleene-4}.
	The distinction between nodes that evaluate to elements of $\mathcal{R}$ and nodes that evaluate to elements
	of $\mathcal{R}[x]$ is directly displayed by the node label.
	Furthermore, if a node $u \in V$ has a descendant $v \in V$ labelled by $\hat *$
	such that all nodes on the path from $u$ to $v$ are labelled by $\circ$, except from $v$ itself,
	then $u$ belongs to $V_2$, otherwise to $V_1$. This allows to define $V_0$, $V_1$, and $V_2$ 
	in FO using the EC-representation of the circuit.
			
	Using a guarded transduction we keep every node in $V_0$, replace every node in $V_1$ by 3 nodes (which compute
	the three regular languages $a,b,c$  in \eqref{eqn:kleene-3})
	and replace every node in $V_2$ by 6 nodes (which compute
	the six regular languages $a,b,c,\alpha,\gamma,\delta$ in \eqref{eqn:kleene-3}). 
	Moreover, we can define the wires accordingly.
			
	Let us consider one specific case (the most difficult one) for the definition of the wires.
	Assume that $A = A_1 \circ A_2$ where $A_1, A_2, A \in V_2$,
	$A_i$ was replaced by the six nodes $a_i,b_i,c_i, \alpha_i,\gamma_i, \delta_i$
	and $A$ was replaced by the six nodes $a,b,c, \alpha,\gamma, \delta$.
	Then $A_i$ computes the term function $t_i(x) =  \alpha(a_i x b_i+c_i)^* \gamma_i + \delta_i$ and
	$A$ has to compute the composition 
	\begin{eqnarray*}
		t(x) = t_2(t_1(x)) & = &  \alpha_2 (a_2    [ \alpha_1 (a_1 x b_1+c_1)^* \gamma_1 + \delta_1 ] b_2+c_2)^* \gamma_2 + \delta_2 \\
		& = &  \alpha_2 \, ( \, \underbrace{a_2 \alpha_1}_{\alpha'} \, ( \, \underbrace{a_1 x b_1+c_1}_{\beta'} \, )^* \, \underbrace{\gamma_1 b_2}_{\gamma} \, + \, 
		\underbrace{a_2 \delta_1 b_2 +c_2}_{\delta'} \, )^*\,  \gamma_2 + \delta_2 
	\end{eqnarray*}
	Using the above identity $(\dagger)$, the expression in the last line becomes equivalent to
	\begin{eqnarray*}
		& &  \alpha_2  \, (  {\delta'}^* \alpha' (\beta' + \gamma' {\delta'}^* \alpha')^* \gamma' {\delta'}^*  + {\delta'}^* ) \, \gamma_2 \, + \, \delta_2  \\
		& = &  \alpha_2 {\delta'}^* \alpha' (\beta' + \gamma' {\delta'}^* \alpha')^* \gamma' {\delta'}^* \gamma_2  + \alpha_2 {\delta'}^* \gamma_2 + \delta_2 \\
		& = &  \underbrace{\alpha_2 {\delta'}^* \alpha'}_{\alpha} (a_1 x b_1+ \underbrace{c_1 + \gamma' {\delta'}^* \alpha'}_{c})^* \underbrace{\gamma' {\delta'}^* \gamma_2}_{\gamma}  + 
		\underbrace{\alpha_2 {\delta'}^* \gamma_2 + \delta_2}_{\delta} .
	\end{eqnarray*}
	Hence, we can define $a = a_1$, $b = b_1$ (these are copy gates) and $c, \alpha, \gamma$, and $\delta$ as shown above.
	For the latter, we have to introduce a constant number of additional gates to built up the above terms for $c, \alpha, \gamma$, and $\delta$.
	For instance, we have 
	$$\alpha = \alpha_2 {\delta'}^* \alpha' = \alpha_2 (a_2 \delta_1 b_2 +c_2)^* \alpha_2 \alpha_1,
	$$
	and we need seven more gates to built up this expression from $\alpha_1, \alpha_1, \delta_1, a_2,b_2, c_2$.
	These seven gates are also produced by the guarded transduction.
	From Lemma~\ref{lem:copy-gates} we obtain the desired circuit over $\mathcal{R}$.
\end{proof}
	
	
A {\em semiring} $\mathcal{R} = (R,+,\cdot)$ is a structure with two associative binary operations $+$ and $\cdot$,
such that $a \cdot (b+c) = a \cdot b + a \cdot c$ and $(a+b) \cdot c = a \cdot c + b \cdot c$ for all $a,b,c \in R$.
Notice we do not require a semiring to have a zero- or a one-element.
Using the same strategy as in the proof of Theorem~\ref{theorem-reg-exp} one can show the following result:
\begin{theorem} \label{theorem-semiring}
	Let $\mathcal{R}$ be a semiring. 
	Given an expression $t$ over $\mathcal{R}$, one can compute in $\TC^0$ an equivalent circuit over $\mathcal{R}$ in EC-representation
	of size $O(n/\log n)$ and depth $O(\log n)$.
\end{theorem}
	
\begin{proof}
	For a semiring $\mathcal{R}$ one has to observe that every linear term function $t^{\mathcal{R}}$ can be written as
	$t^{\mathcal{R}}(x) = a x b + c$ for semiring elements $a,b,c \in \mathcal{R}$ where any of the elements $a,b,c$ can also be missing.  
	In other words, the right-hand side of $t^{\mathcal{R}}(x)$ can be of one of the following 8 forms for $a,b,c \in \mathcal{R}$:
	$a \cdot x \cdot b+c$, $a \cdot x \cdot b$, $a \cdot x+c$, $x \cdot b+c$, $a \cdot x$, $x \cdot b$, $x+c$, $x$.
	By Theorem~\ref{thm:main-algebra} we compute from the input semiring expression $t$ in $\TC^0$ an equivalent circuit $\C$ over
	$\mathcal{R}[x]$ in EC-representation. We partition $V(\C)$ into $V(\C) = V_0 \cup V_1 \cup \dots \cup V_8$
	where $V_0$ contains all gates which evaluate to a semiring element and $V_1 \cup \dots \cup V_8$ contain
	gates which evaluate to a linear term function, grouped by the 8 possible types listed above.
	It is easy to see that the sets $V_i$ are FO-definable using the EC-representation of $\C$.
	For example, a gate $v$ carries an $a$-coefficient, i.e. it computes a term function of the form $a \cdot x$, $a \cdot x \cdot b$, $a \cdot x+c$ or $a \cdot x \cdot b+c$, if and only if there exists a path $v = v_1, v_2, \dots, v_m$
	such all nodes $v_1, \dots, v_{m-1}$ are labelled by the binary $\mathcal{F}(\mathcal{R})$-operation $\circ$ and $v_m$ is labelled by
	the unary $\mathcal{F}(\mathcal{R})$-operation ${\hat \cdot}_2$ that maps $a \in \mathcal{R}$ to the linear term function $a \cdot x$.
	This allows to carry over the arguments from the
	proof of Theorem~\ref{theorem-reg-exp}.
\end{proof}
	
\section{Future work} 
	
In the recent paper \cite{GanardiHKL17}, we proved a dichotomy theorem for the circuit evaluation problem for non-commutative semirings (which are not required
that have a $0$ or $1$): the problem is in $\DET \subseteq \NC^2$ or $\P$-complete. Moreover, we gave precise algebraic
characterizations for the two corresponding classes of semirings. In a forthcoming paper, we will prove a similar dichotomy theorem for 
the expression evaluation problem for a non-commutative semiring: it is in $\TC^0$ or $\NC^1$-complete. Theorem~\ref{theorem-semiring}	
will play a crucial role in the proof. This shows that our $\TC^0$-balancing procedure can be also used to show that a problem belongs to $\TC^0$,
despite the fact that the circuits we produce have logarithmic depth and not constant depth.
	
An interesting problem in connection with Theorem~\ref{thm:main-algebra} is to determine 
further classes of algebras $\A$ for which one can compute in $\TC^0$ for a given term over $\A$
an equivalent term over $\A$ of logarithmic depth.
	
\bibliographystyle{plain}

\end{document}